%% file: main.tex
\documentclass{article}

\usepackage{arxiv}

\usepackage[utf8]{inputenc} % allow utf-8 input
\usepackage[T1]{fontenc}    % use 8-bit T1 fonts
\usepackage{hyperref}       % hyperlinks
\usepackage{url}            % simple URL typesetting
\usepackage{booktabs}       % professional-quality tables
\usepackage{amsfonts}       % blackboard math symbols
\usepackage{nicefrac}       % compact symbols for 1/2, etc.
\usepackage{microtype}      % microtypography
\usepackage{lipsum}
\usepackage{amsmath}

\usepackage[utf8]{inputenc}
\usepackage[english]{babel}

\usepackage{blindtext}
\usepackage{amssymb}

\newtheorem{theorem}{Theorem}[section]

\newtheorem{lemma}[theorem]{Lemma}
\newtheorem{proof}[theorem]{proof}

\usepackage{graphicx}
\usepackage{balance}  % for  \balance command ON LAST PAGE  (only there!)
\usepackage{xcolor}

\usepackage{balance}
\usepackage{graphicx, comment}
\usepackage{cite}
\usepackage{url}

\usepackage{breakurl}
\usepackage{balance}  % for  \balance command ON LAST PAGE  (only there!)
\usepackage{booktabs} % For formal tables
\usepackage{tabularx}
\usepackage{algorithm, algorithmicx, algpseudocode}
\usepackage{pifont}
\usepackage{xspace}
\usepackage{caption}
\usepackage{subcaption}

\newcommand{\protocolwitness}{AC\textsuperscript{3}WN\xspace}
\newcommand{\protocolname}{AC\textsuperscript{3}\xspace}

\newcommand{\transactionname}{AC\textsuperscript{2}T\xspace}

\usepackage{titlesec}
\titlespacing\section{1pt}{2pt plus 2pt minus 2pt}{2pt plus 2pt minus 2pt}
\titlespacing\subsection{1pt}{2pt plus 4pt minus 2pt}{2pt plus 2pt minus 2pt}
\titlespacing\subsubsection{0pt}{2pt plus 4pt minus 2pt}{2pt plus 2pt minus 2pt}

\title{Atomic Commitment Across Blockchains}

\author{
  Victor Zakhary \\
  Department of Computer Science\\
  UC Santa Barbara\\
  Santa Barbara, CA 93106 \\
  \texttt{victorzakhary@ucsb.edu} \\
   \And
  Divyakant Agrawal \\
  Department of Computer Science\\
  UC Santa Barbara\\
  Santa Barbara, CA 93106 \\
  \texttt{divyagrawal@ucsb.edu} \\
   \And
  Amr El Abbadi \\
  Department of Computer Science\\
  UC Santa Barbara\\
  Santa Barbara, CA 93106 \\
  \texttt{elabbadi@ucsb.edu} \\
}

\begin{document}
\maketitle

\begin{abstract}
The recent adoption of blockchain technologies and open 
permissionless networks suggest the importance of 
peer-to-peer atomic cross-chain transaction protocols. 
Users should be able to atomically exchange tokens and 
assets without depending on centralized 
intermediaries such as exchanges. Recent peer-to-peer atomic 
cross-chain swap protocols use hashlocks and timelocks
to ensure that participants comply to the protocol. However,
an expired timelock could lead to a violation of the all-or-nothing
atomicity property. An honest participant who fails to execute a smart 
contract on time due to a crash failure or network delays at her site 
might end up losing her assets. Although a crashed participant
is the only participant who ends up worse off, current proposals are
unsuitable for atomic cross-chain transactions in asynchronous 
environments where crash failures and network delays are the norm. 
In this paper, we present \protocolwitness, the 
first decentralized all-or-nothing \textit{atomic} cross-chain commitment 
protocol. The redeem and refund events of the smart contracts that exchange assets are 
modeled as conflicting events. An open permissionless network of witnesses is 
used to guarantee that conflicting events could never simultaneously occur 
and either all smart contracts in an atomic cross-chain transaction
are redeemed or all of them are refunded.
\end{abstract}

% keywords can be removed

\input{introduction.tex}

\input{background.tex}

\input{model.tex}

\input{solution.tex}

\input{generalization.tex}

\input{evaluation.tex}

\input{conclusion.tex}

\balance
\bibliographystyle{abbrv}
\bibliography{main}
\end{document}

%% file: introduction.tex
\section{Introduction} \label{sec:introduction}

% The need for cross-chain transactions
The wide adoption of permissionless open blockchain networks by both industry (e.g., 
Bitcoin~\cite{nakamoto2008bitcoin}, Ethereum~\cite{wood2014ethereum}, etc) 
and academia (e.g., Bzycoin~\cite{kogias2016enhancing}, Elastico~\cite{luu2016secure}, 
BitcoinNG~\cite{eyal2016bitcoin}, Algorand~\cite{micali2016algorand}, etc) suggest the 
importance of developing protocols and infrastructures that support peer-to-peer atomic
cross-chain transactions. Users, who usually do not trust each other,
should be able to directly exchange their tokens 
and assets that are stored on different blockchains (e.g., Bitcoin and Ethereum)
without depending on trusted third party intermidiaries.
Decentralized permissionless~\cite{maiyya2018database} blockchain
ecosystems require infrastructure enablers and protocols that allow users to 
atomically exchange tokens without giving up trust-free decentralization,
the main reasons behind using  permissionless blockchain. We motivate the problem
of atomic cross-chain transactions and discuss
the current available solutions and their limitations through the following example.

\begin{comment}
depending on centralized intermediaries 
such as exchanges. 
{\bf I would put sharding as a benefit later, not so early.}In addition, these protocols for peer-to-peer atomic cross-chain 
commitment can be used to increase the transaction throughput of a blockchain network 
through \textit{sharding}~\cite{herlihy2018atomic, onsharding2018}. Different network nodes
can independently be responsible for different partitions of the blockchain without 
coordinating on every single transaction. Only transactions that span multiple partitions 
are required to coordinate using an atomic cross-chain or cross-partition swap protocols.
\end{comment}

% Centralized solution through exchanges

Suppose Alice owns X bitcoins and she wants to exchange them for Y 
ethers. Luckily, Bob owns ether and he is willing to exchange his Y ethers for X bitcoins. 
In this example, Alice and Bob want 
to atomically exchange assets that reside in different blockchains. In addition,
both Alice and Bob \textbf{do not trust} each other and in many scenarios, 
they might not be co-located to do this atomic exchange in person.
Current infrastructures
do not support these direct peer-to-peer transactions. 
Instead, both Alice and Bob need to \textbf{independently} exchange their 
tokens through a trusted centralized exchange, Trent 
(e.g., Coinbase~\cite{coinbase} and Robinhood~\cite{robinhood}) either through
fiat currency or directly.  Using Fiat, both Alice
and Bob first exchange their tokens with Trent for a fiat currency (e.g., USD) and 
then use the earned fiat currency to buy the other token also from Trent or from 
another trusted exchange. Alternatively, some exchanges (e.g., Coinbase) allow their 
customers to directly exchange tokens (ether for bitcoin or bitcoin for ether) 
without going through fiat currencies.

%A straight forward optimization to the presented model in 
% Figure~\ref{fig:current_system}
%is shown in Figure~\ref{fig:current_system_optimized}. Trent, the trusted 
% centralized 
% exchange, allows Alice and Bob to directly exchange bitcoin for ehter without 
% going through fiat currencies.

% As shown in 
%Figure~\ref{fig:current_system}, current infrastructures do not allow Alice and 
%Bob to directly swap their tokens and resources. Instead, 

These solutions have many drawbacks that make them unacceptable
solutions for atomic peer-to-peer cross-chain transactions. \textit{First}, 
both solutions require both Alice and Bob to trust Trent. This 
centralized trust requirement risks to derail the whole idea of 
blockchain's trust-free decentralization~\cite{nakamoto2008bitcoin}. 
\textit{Second}, both solutions require Trent to trade in all 
involved resources (e.g., bitcoin and ether). This requirement is 
unrealistic especially if Alice and Bob want to exchange commodity 
resources (e.g., transfer a car ownership for bitcoin assuming car titles
are stored in a blockchain~\cite{herlihy2018atomic, zakhary2019towards}). \textit{Third}, 
both solutions do not achieve atomicity of the transaction among the involved 
participants. Alice might trade her bitcoin directly for 
ether or through a fiat currency while Bob has no obligation to 
execute his part of the swap. Finally, both solutions significantly
increase the number of required transactions to achieve the intended cross-chain
transaction, and hence drastically increases the imposed fees. One cross-chain 
transaction between Alice and Bob results in either four transactions (two between Alice
and Trent and two between Bob and Trent) if fiat is used or at best
two transactions (one between Alice and Trent and one between Bob and Trent)
if assets are directly swapped.

An \textbf{A}tomic \textbf{C}ross-\textbf{C}hain \textbf{T}ransaction\footnote{Atomic Cross-Chain Transaction, \transactionname, and Atomic Swap, AS, are interchangeably used.}, \transactionname,
is a distributed transaction that 
spans multiple blockchains. This distributed transaction consists of
sub-transactions and each sub-transaction is executed on some blockchain.
An \textbf{A}tomic \textbf{C}ross-\textbf{C}hain \textbf{C}ommitment, \protocolname,
%\footnote{Atomic commitment and atomic swap are 
% interchangeably used.} 
protocol is required to execute AC$^2$Ts. This protocol is a variation of traditional
distributed atomic commitment protocols (e.g., 2PC~\cite{gray1978notes, bernstein1987concurrency}). This protocol 
should guarantee both \textit{atomicity} and \textit{commitment}
of AC$^2$Ts. 
\textbf{Atomicity} ensures the 
\textbf{all-or-nothing} property where either all sub-transactions
take place or none of them is executed. \textbf{Commitment} guarantees 
that any changes caused by a cross-chain transaction must eventually take
place if the transaction is decided to commit. Unlike in 2PC and other traditional distributed 
atomic commitment protocols, atomic cross-chain commitment protocols 
are also trust-free and therefore 
must {\bf tolerate}  maliciousness~\cite{herlihy2018atomic}. 
A two-party atomic cross-chain swap protocol was originally proposed by 
Nolan~\cite{atomicNolan, atomicTrading} and generalized by 
Herlihy~\cite{herlihy2018atomic} to process multi-party atomic cross-chain swaps.
Both Nolan's protocol and its generalization by Herlihy use smart contracts, hashlocks, 
and timelocks to execute atomic cross-chain swaps. A smart contract is a self executing
contract (or a program) that gets executed in a blockchain 
once all the terms of the contract are satisfied. A hashlock is 
a cryptographic one-way hash function $h = H(s)$ that locks
assets in a smart contract until a hash secret $s$ is provided. A timelock is a time 
bounded lock that triggers the execution of a smart contract function
after a pre-specified time period. 

The atomic swap between Alice and Bob, explained in the earlier example,
is executed using Nolan's protocol as follows. Let a participant
be the leader of the swap, say Alice. Alice creates a 
secret $s$, only known to Alice, and a hashlock $h = H(s)$. 
Alice uses $h$ to lock X bitcoins
 in a smart contract $SC_1$ and publishes $SC_1$ in the Bitcoin network. 
$SC_1$ states to transfer X bitcoins to Bob if Bob provides the 
secret $s$ such that $h = H(s)$ to $SC_1$. 
In addition, $SC_1$ is locked with a timelock $t_1$ that refunds the X bitcoins 
to Alice if Bob fails to provide $s$ to $SC_1$ before $t_1$ expires. As $SC_1$ is published
in the Bitcoin network and made public to everyone, Bob can verify that $SC_1$ indeed
transfers X bitcoins to the public address of Bob if Bob provides $s$ to $SC_1$. 
In addition, Bob learns $h$ from $SC_1$. Using $h$, Bob publishes a smart contract 
$SC_2$ in the Ethereum network that locks Y ethers in $SC_2$ using $h$. $SC_2$ 
states to transfer Y ethers to Alice if Alice provides the secret $s$ to $SC_2$. 
In addition, $SC_2$ is locked with a timelock $t_2 < t_1$ that refunds the Y ethers 
to Bob if Alice fails to provide $s$ to $SC_2$ before $t_2$ expires. 

Now, if Alice wants to redeem her Y ethers from $SC_2$, Alice must
reveal $s$ to $SC_2$ before $t_2$ expires. Once $s$ is provided to
$SC_2$, Alice redeems the Y ethers and $s$ gets revealed to Bob.
Now, Bob can use $s$ to redeem his X bitcoins from $SC_1$ before
$t_1$ expires. Notice that $t_1 > t_2$ is a necessary condition to ensure that Bob has enough time to redeem his X bitcoins from $SC_1$
after Alice provides $s$ to $SC_2$ and before $t_1$ expires. 
If Bob provides $s$ to $SC_1$ before $t_1$ expires, Bob successfully 
redeems his X bitcoins and the atomic swap is marked completed.

\textbf{The case against the current proposals:} if Bob fails to 
provide $s$ to $SC_1$ before $t_1$ expires due to a crash failure or 
a network partitioning at Bob's site, Bob loses his X bitcoins and 
$SC_1$ refunds the X bitcoins to Alice. This violation of the atomicity 
property of the protocol penalizes Bob for a failure that happens out 
of his control. Although a crashed participant is the only participant 
who ends up being worse off (Bob in this example), this protocol does 
not guarantee the atomicity of AC$^2$Ts in asynchronous environments 
where crash failures, network partitioning, and message delays are the 
norm.

\begin{sloppypar}
Another important drawback in Nolan's and Herlihy's protocols
is the requirement of sequentially publishing the smart contracts in an atomic swap before the leader (Alice in our 
example) reveals the secret $s$. This requirement is necessary to ensure that
the publishing events of all the smart contracts in the atomic
swap \textit{happen before} the redemption of any of the smart contracts. 
This causality requirement ensures that any malicious participant
who declines to publish a smart contract does not take advantage of the protocol.
However, the sequential publishing of smart contracts, especially in atomic swaps
that include many participants, proportionally increases the latency of the swap 
to the number of sequentially published contracts.
\end{sloppypar}

In this paper, we propose \textbf{\protocolwitness}, the first 
decentralized all-or-nothing 
\textbf{A}tomic \textbf{C}ross-\textbf{C}hain \textbf{C}ommitment protocol that uses an open \textbf{W}itness \textbf{N}etwork. The redemption and the refund events of smart contracts in \transactionname are modeled as conflicting events. 
A decentralized open network of witnesses is used to guarantee that conflicting events must
never simultaneously take place and 
either all smart contracts in an \transactionname are redeemed or all of them are refunded. Unlike in Nolan's and Herlihy's protocols, 
\protocolwitness allows all participants to concurrently
publish their contracts in a swap resulting in a drastic
decrease in an atomic swap's latency. Our contribution is 
summarized as follows:
\begin{itemize}

    \item We present \protocolwitness, the first all-or-nothing atomic
    cross-chain commitment protocol. \protocolwitness is decentralized and does not
    require to trust any centralized intermediary.
    \item We prove the correctness of \protocolwitness showing that
    \protocolwitness achieves both atomicity and commitment of AC$^2$Ts.
    
    \item Finally, we analytically evaluate \protocolwitness in comparison
    to Herlihy's~\cite{herlihy2018atomic} protocol. Unlike in Herlihy's protocol
    where the latency of an atomic swap proportionally increases as the number
    of the sequentially published smart contracts in the atomic swap increases, 
    our analysis shows that the latency of an atomic swap in 
    \protocolwitness is constant irrespective of the number of smart contracts involved.
\end{itemize}

\begin{comment}
In addition, our protocol addresses 
the problem of forks in blockchains to ensure that both \textit{atomicity} and 
\textit{commitment} are achieved. 
\end{comment}

The rest of the paper is organized as follows. In Section~\ref{sec:background}, we 
discuss the open blockchain data and transactional models. 
Section~\ref{sec:model} explains the cross-chain distributed
transaction model and Section~\ref{sec:solution}
presents our atomic cross-chain commitment protocol. An analysis of the atomic
cross-chain commitment protocol is
presented in Section~\ref{sec:generalizations}. The protocol is evaluated in 
Section~\ref{sec:evaluation} and the paper is concluded in Section~\ref{sec:conclusion}.

%% file: background.tex
\section{Open Blockchain Models} \label{sec:background}

\subsection{Architecture Overview}
An open permissionless blockchain system~\cite{maiyya2018database} 
(e.g., Bitcoin, Ethereum) typically consists of two layers: a storage layer
and an application layer as illustrated in Figure~\ref{fig:architecture}. 
\textbf{The storage layer} comprises a decentralized 
distributed ledger managed by an open network of computing 
nodes. A blockchain system is permissionless if computing nodes can join or 
leave the network of its storage layer at any moment
without obtaining a permission from a centralized authority. Each computing 
node, also called a miner, maintains a copy of the ledger. The ledger is a tamper-proof 
chain of blocks, hence named blockchain. Each block contains a set of 
valid transactions that transfer assets among end-users. \textbf{The application layer}
comprises end-users who communicate with the storage layer via \textit{message 
passing} through a client library. End-users have identities, defined by their public keys, and signatures, generated using their private keys. 
Digital signatures are the end-users'
way to generate transactions as explained later in 
Section~\ref{sub:transaction_model}. End-users submit their transactions
to the storage layer through a client library. Transactions are used to 
transfer assets from one identity to another. End-users multicast their 
transaction messages to mining nodes in the storage layer.
\begin{figure}[ht!]
	\centering
    \includegraphics[width=0.5\columnwidth]{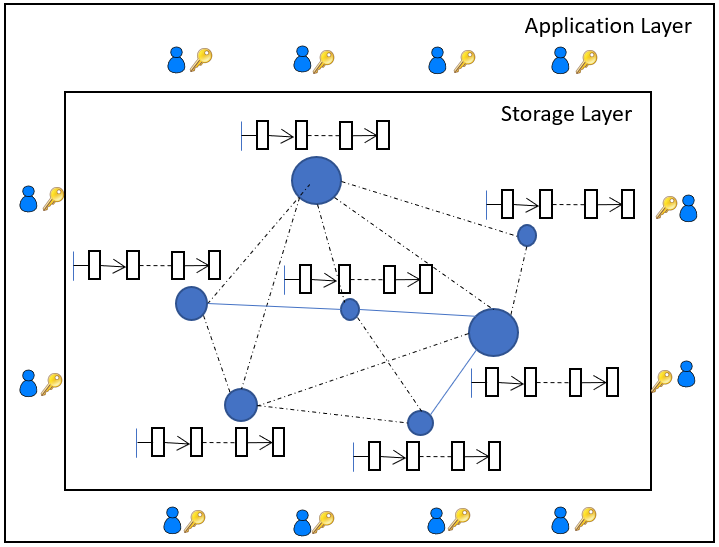}
    \caption{Open blockchain architecture overview.}
    \label{fig:architecture}
\end{figure}

A mining node validates the transactions it receives and valid transactions
are added to the current block of a mining node. Miners run a consensus
protocol through mining to agree on the next block to be added to the chain.
A miner who mines a block gets the right to add its block to the chain and 
multicasts it to other miners. To make progress, miners accept the first 
received mined block after verifying it and start mining the 
next block\footnote{Forks and fork resolutions are discussed in later Sections.}. 
Sections~\ref{sub:data_model} and~\ref{sub:transaction_model} explain the data model
and the transactional model of open blockchain systems respectively.

\begin{comment}
These transactions transfer the ownership of an asset from the sender identity
to 
Submitted transactions are validated by the miners of the storage layer.
Miners verified not to conflict with any other transaction in the same block or any 
previous block. 
End-users submit transactions to miners and miners verify these
transactions and add them to their current block. Miners run a consensus
protocol through mining to agree on the next block to be added to the chain.
A miner who mines a block gets the right to add this block to the chain and 
multi-cast it to other miners. To make progress, miners accept the first 
received block after verifying it and start mining the next block. 
\end{comment}

\subsection{Data Model} \label{sub:data_model}
The storage layer stores the ownership information of assets
in the system in the blockchain.  The ownership is determined through 
identities and identities are typically implemented using public keys.
For example, the Bitcoin blockchain stores the 
information of the most recent owner of every bitcoin in the 
Bitcoin blockchain.  A bitcoin 
that is linked to Alice's public key is owned by Alice. In addition,
the blockchain stores transactions that transfer the ownership of an
asset from one identity to another. Therefore, an asset can be tracked 
from its registration in the blockchain, the first owner, to its last owner 
in the blockchain. In the Bitcoin
network, new bitcoins are generated and registered in the Bitcoin blockchain
through mining. Asset ownership is transferred from one identity to 
another through a transaction. In addition, transactions are used to merge 
or split assets as explained in Section~\ref{sub:transaction_model}.

\subsection{Transaction Model} \label{sub:transaction_model}

\begin{figure}[ht!]
	\centering
    \includegraphics[width=\columnwidth]{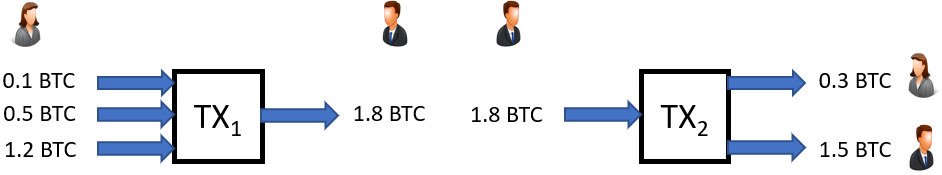}
    \caption{Asset transactional model.}
    \label{fig:transations}
\end{figure}
A transaction is a digital signature 
that transfers the ownership of assets from one identity
to another. End-users, in the application layer,
use their private keys~\cite{rivest1978method} to digitally sign 
assets linked to their identity to transfer these assets to other 
identities, identified by their public keys. These digital signatures
are submitted to the storage layer via message passing through a 
client library. It is the responsibility of the miners to validate 
that end-users can transact only on their own assets. If an 
end-user digitally signs an asset that is not owned by this end-user,
the resulting transaction is not valid and is rejected by the miners.
In addition, miners validate that an asset cannot be spent twice and hence
prevent double spending of assets.

A transaction takes one or more input assets owned by one identity 
and results in one or more output assets where each output asset
is owned by one identity. Therefore,
transactions are used to merge or split assets. Figure~\ref{fig:transations}
shows two transactions $TX_1$, a transaction that merges assets, 
and $TX_2$, a transaction that splits assets.  $TX_1$ takes 3 input assets
owned by Alice, merges them into one output asset, and transfers the 
ownership of this merged asset to Bob. On the other hand, $TX_2$ takes 
one input asset owned by Bob and splits it into 2 output assets of two 
different values; one is transferred to Alice and the other is transferred 
to Bob. Note that the summation of a transaction's input assets matches the 
summation of its output assets assuming that no transaction fees are 
imposed.

\begin{figure}[ht!]
	\centering
    \includegraphics[width=\columnwidth]{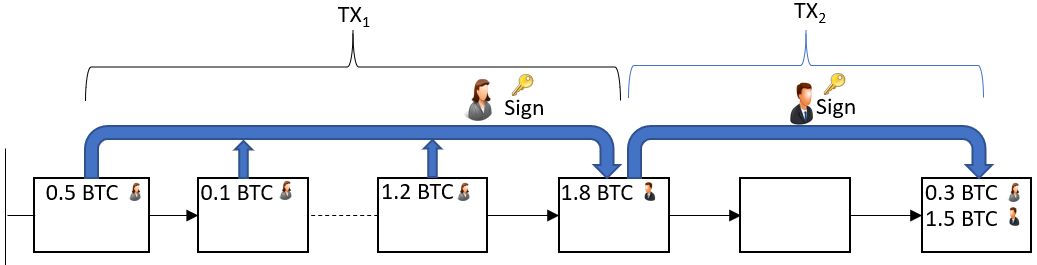}
    \caption{The blockchain representation of $TX_1$ and $TX_2$ of Figure~\ref{fig:transations}.}
    \label{fig:blockchain_transactions}
\end{figure}

Figure~\ref{fig:blockchain_transactions} shows an example of how $TX_1$
and $TX_2$ take place in the Bitcoin blockchain. As shown, Alice can
 only transact on assets that she owns in previous blocks in the blockchain
issuing $TX_1$. Similarly, once the ownership
of 1.8 bitcoin is transferred to Bob, only then can Bob issue the 
transaction
$TX_2$ to split the 1.8 bitcoin asset to 0.3 to Alice and 1.5 to Bob 
in a following block. In traditional databases, 
end-user transactions execute arbitrary updates in the storage
layer as long as the semantic and the access control rights 
of a transaction are validated in the 
application layer. On the other hand, 
in blockchain systems, this validation is 
explicitly enforced in the storage layer and hence end-users, in the 
application layer, are allowed to transact only on assets
they own in the storage layer.

Another way to perform transactions in blockchain systems is through \textbf{smart
contracts}. A smart contract is a program written in some scripting language
(e.g., Solidity for Ethereum smart contracts~\cite{solidity}) that
allows general program executions on a blockchain's mining nodes.
End-users deploy a smart contract in a blockchain through a deployment\footnote{Deployment and publishing are used interchangeably.} message, 
$msg$, that is sent to the mining nodes in the storage layer. 
The deployment message includes the smart contract code in addition to 
some implicit parameters that are accessible to the smart contract code once
the smart contract is deployed. These parameters include the sender 
end-user public key, accessed through $msg.sender$, and an optional asset 
value, accessed through $msg.val$. This optional asset value allows end-users to 
lock some of their assets in the deployed smart contract. Like transactions, a smart
contract is deployed in a blockchain if it is included in a mined block in 
this blockchain. We adopt Herlihy's notion 
of a smart contract as an object in programming 
languages~\cite{herlihy2019blockchains,dickerson2017adding}.
A smart contract has a state, a constructor that is called when 
a smart contract is first deployed in the blockchain, and a set
of functions that could alter the state of the smart contract. The 
constructor initializes the smart contract's state and uses
the implicit parameters to initialize the owner of the smart contract 
and the assets to be locked in this smart contract. 
Miners verify that the end-user who deploys 
a smart contract indeed owns these assets. Once assets are locked in a 
smart contract, their owners cannot transact on these assets outside the 
smart contract logic until these assets are unlocked from the smart contract 
as a result of a smart contract function call. To execute
a smart contract function, end-users submit their function call accompanied
by the function parameters through messages to miners. These messages could include 
implicit parameters as well (e.g., $msg.sender$). Miners 
execute\footnote{End-users pay to miners a smart contract deployment fee
plus a function invocation fee for every function call.} the function on the current
contract state and record any contract state changes in their current block in the 
blockchain. Therefore, a smart contract state might span many blocks after the block where the smart contract is first deployed.

%% file: model.tex
\section{Atomic Cross-Chain Transaction Model}\label{sec:model}

\begin{figure}[ht!]
	\centering
    \includegraphics[width=0.4\columnwidth]{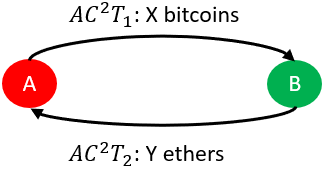}
    \caption{An atomic cross-chain transaction graph to swap X bitcoins for Y ethers
    between Alice (A) and Bob (B).}
    \label{fig:atomic-swap-graph}
\end{figure}

An Atomic Cross-Chain Transaction, \transactionname,  is a distributed transaction to 
transfer the ownership of assets stored in multiple blockchains among two or 
more participants. This distributed transaction consists of
sub-transactions and each sub-transaction transfers an asset on some blockchain.
An \transactionname is modeled using a directed graph $\mathcal{D} = 
(\mathcal{V},\mathcal{E})$~\cite{herlihy2018atomic} where
$\mathcal{V}$ is the set of vertexes and $\mathcal{E}$ is the set of 
edges in $\mathcal{D}$. $\mathcal{V}$ represents the participants in \transactionname
and $\mathcal{E}$ represents the sub-transactions in \transactionname. A directed edge $e = 
(u, v) \in \mathcal{E}$ represents a sub-transaction that transfers an asset $e.a$ from a source
participant $u \in \mathcal{V}$ to a recipient participant $v \in 
\mathcal{V}$ in some blockchain $e.BC$. Figure~\ref{fig:atomic-swap-graph} shows an example
of an \transactionname graph between Alice (A) and Bob (B). As shown,
the edge (A, B) represents the sub-transaction $AC^{2}T_1$ that
transfers X bitcoins from A to B while the edge (B, A) represents
the sub-transaction $AC^{2}T_2$ that transfers Y ethers from B to A.

An atomic cross-chain commitment protocol is required in order to 
correctly execute an \transactionname. This protocol must ensure the 
atomicity and the commitment of all sub-transactions in \transactionname as 
follows.

\begin{itemize}
    \item Atomicity: either all asset transfers of all sub-transactions in 
    the \transactionname take place or none of them does.
    \item Commitment: once the atomic cross-chain commitment protocol decides
    the commitment of an \transactionname, all asset 
    transfers of all sub-transactions in this \transactionname must eventually take place.
\end{itemize}

\begin{comment}
An atomic cross-chain commitment protocol is a variation of 2PC. Therefore,
we use the analogy of 2PC to explain an abstraction of an
atomic cross-chain commitment protocol. In 2PC, a distributed transaction
spans multiple data partitions and each partition is responsible for
executing a sub-transaction. A coordinator sends
a vote request to all involved data partitions. Upon receiving a vote
request, a data partition locks the data objects accessed by its
sub-transaction and votes back \textit{yes} only if it succeeds to
lock all the involved data objects. This locking represents an escrow
that this sub-transaction can be successfully executed on this 
data partition in isolation from any concurrent transaction running 
on the same data partition. If the locking of some data objects fails,
a data partition votes \textit{no} to the coordinator. A coordinator
decides to commit a distributed transaction if all involved data 
partitions vote yes, otherwise it decides to abort the distributed
transaction. If a commit decision is reached, all data partitions
execute their sub-transactions and release the locks. However, if an 
abort is decided, data partitions abort their sub-transactions and 
release the locks as well. 2PC assumes that the coordinator and the data
partitions are trusted. The main challenge in blockchain systems is
how to design a trust-free variant analog of 2PC where end-user participants do 
not trust each other and a protocol cannot depend on a centralized trusted 
coordinator.

\end{comment}

An atomic cross-chain commitment protocol is a variation of the two phase commit protocol (2PC)~\cite{gray1978notes, bernstein1987concurrency}. Therefore,
we use the analogy of 2PC  to explain an abstraction of an
atomic cross-chain commitment protocols. In 2PC, a distributed transaction
spans multiple data partitions and each partition is responsible for
executing a sub-transaction. A coordinator sends
a vote request to all involved data partitions. Upon receiving a vote
request, a data partition votes back \textit{yes} only if it succeeds 
in executing all the operations of its sub-transaction on the involved data objects. 
Otherwise, a data partition votes \textit{no} to the coordinator. A coordinator
decides to commit a distributed transaction if all involved data 
partitions vote yes, otherwise it decides to abort the distributed
transaction. If a commit decision is reached, all data partitions
commit their sub-transactions. However, if an 
abort decision is reached, data partitions abort their sub-transactions. 
2PC assumes that the coordinator and the data
partitions are trusted. The main challenge in blockchain systems is
how to design a trust-free variation of 2PC where end-user participants do 
not trust each other and a protocol cannot depend on a centralized trusted 
coordinator.

An atomic cross-chain commitment protocol requires that for 
every edge $e = (u, v) \in \mathcal{E}$, the source participant $u$ 
to lock an asset $e.a$ in Blockchain $e.BC$. This asset locking is 
necessary to temporarily prevent the participant $u$ from spending 
$e.a$ through other transactions in $e.BC$. If every source participant 
$u$ locks $e.a$ in $e.BC$, the atomic cross-chain commitment protocol can 
decide to commit the \transactionname. Once the protocol decides to 
commit the \transactionname, every recipient participant $v$ should be 
able to \textit{redeem} the asset $e.a$. However, if the protocol decides to 
abort the \transactionname because some participants do not comply to the 
protocol or a participant requests the transaction to abort, every source 
participant $u$ should be able to refund their locked assets $e.a$.

In blockchain systems, smart contracts are used to implement this logic. 
Participant $u$ deploys a smart contract $SC_e$ 
in Blockchain $e.BC$ to lock an asset $e.a$ owned by $u$ in $SC_e$.
$SC_e$ ascertains to conditionally transfer $e.a$ to $v$ if a commitment 
decision is reached, otherwise $e.a$ is refunded to $u$. A smart contract 
$SC_e$ exists in one of three states: \textit{published} ($P$), 
\textit{redeemed} ($RD$), or \textit{refunded} ($RF$). A smart 
contract $SC_e$ is \textit{published} if it gets deployed to $e.BC$ by $u$. 
Publishing the smart contract $SC_e$ serves \textit{two} important goals towards 
the atomic execution of an \transactionname. First, it represents a 
\textit{yes} vote on the sub-transaction corresponding to the 
edge $e$. Second, it locks the asset $e.a$ in blockchain $e.BC$.
A smart contract $SC_e$ is \textit{redeemed} if participant $v$ 
successfully redeems the asset $e.a$ from $SC_e$. Finally, a smart contract $SC_e$ is 
refunded if the asset $e.a$ is refunded to participant $u$.

Now, if for every edge $e = (u, v) \in \mathcal{E}$, participant
$u$ publishes smart contract $SC_e$ in $e.BC$,
it means that all participants vote yes on \transactionname,
lock their involved assets in \transactionname, and hence
the \transactionname can be committed. However, if some participants
decline to publish their smart contracts, the \transactionname has to be aborted.
The commitment of \transactionname requires the redemption of 
\textbf{every} smart contract $SC_e$ in \transactionname. On the other hand, if the 
\transactionname aborts, this requires the refund
of \textbf{every} smart contract $SC_e$ in \transactionname. 

To implement conditional smart contract redemption and refund, 
a cryptographic commitment scheme primitive based on~\cite{goldreich2001foundations}
is used. A \textit{commitment scheme} allows a user to commit to some 
chosen value without revealing this value. Once this hidden value is 
revealed, other users can verify that the revealed value is indeed the 
one that is used in the commitment. A \textit{hashlock} is an example of a 
commitment scheme. A hashlock is a cryptographic one-way hash function $h 
= H(s)$ that is used to conditionally lock assets in a smart contract 
using $h$, the lock, until a hash secret $s$, the key, is revealed. 
Once $s$ is revealed, everyone can verify that the lock $h$ equals 
to $H(s)$ and hence unlocks the assets 
locked in the smart contract. 

An atomic cross-chain commitment protocol should ensure that smart
contracts in \transactionname are \textit{either all redeemed or all refunded}.
For this, a protocol uses \textit{two} mutually exclusive commitment 
scheme instances: a redemption commitment scheme and a refund commitment scheme.
All smart contracts in \transactionname commit their redemption action 
to the redemption commitment scheme instance and their refund action to the refund 
commitment scheme instance. If the protocol decides to commit the \transactionname, 
the protocol must publish the redemption commitment scheme secret. This 
allows all participants in \transactionname to redeem their assets. However, if  
the protocol reaches an abort decision, the protocol must publish
the refund commitment scheme secret. This allows participants 
in \transactionname to refund the locked assets in every published 
smart contract. A protocol must ensure that once the
secret of one commitment scheme instance is revealed, the secret of the other 
instance cannot be revealed. This guarantees the \textit{atomic} execution
of an \transactionname. In Section~\ref{sec:solution}, we instantiate different
protocols that implement \textit{mutually exclusive redemption and refund}
commitment schemes in different ways.

Algorithm~\ref{algo:smart_contract_template} illustrates a smart contract
template that 
can be used in implementing an atomic cross-chain commitment 
protocol. Each smart contract has a sender $s$ and recipient $r$ 
(Line~\ref{line:addresses}), an asset $a$ (Line~\ref{line:asset}) to 
be transferred from $s$ to $r$ through the contract, 
a state (Line~\ref{line:state}), and a redemption and refund commitment scheme
instances $rd$ and $rf$ (Lines~\ref{line:cs_rd} and~\ref{line:cs_rf}). A smart contract
is published in a blockchain through a deployment message. When published, 
its constructor (Line~\ref{line:t-constructor}) is executed to initialize the contract. 
The deployment message of a smart contract typically includes some 
implicit parameters like the sender's address (msg.sender, 
Line~\ref{line:implicit_sender}) and the 
asset value (msg.value, Line~\ref{line:implicit_value}) 
to be locked in the contract. The constructor initializes
the addresses, the asset value, the refund and redemption commitment
schemes, and sets the contract state to P 
(Lines~\ref{line:sender-recceiver}--\ref{line:published}).

\begin{algorithm}[!h]
\caption{An atomic swap smart contract template.}
 \label{algo:smart_contract_template}

abstract class AtomicSwapSC \{
\begin{algorithmic}[1]
    \State enum State \{Published (P), Redeemed (RD), Refunded (RF)\} \label{line:states}
    \State Address s, r // Sender and recipient public keys.
    \label{line:addresses}
    \State Asset a
    \label{line:asset}
    \State State state
    \label{line:state}
    \State CS rd // Redemption commitment scheme
    \label{line:cs_rd}
    \State CS rf // Refund commitment scheme
    \label{line:cs_rf}
    
    \Procedure{Constructor}{Address r, CS rd, CS rf} \label{line:t-constructor}
        \State this.s = msg.sender, this.r = r         \label{line:sender-recceiver}
        \label{line:implicit_sender}
        \State this.a = msg.value
        \label{line:implicit_value}
        \State this.rd = rd, this.rf = rf
        \State state = P
        \label{line:published}
    \EndProcedure
    
    \Procedure{Redeem}{Secret $s_{rd}$} \label{line:t-redeem}
        \State requires(state == P and IsRedeemable($s_{rd}$))
        \label{line:redeem_requirement}
        \State transfer a to r        \label{line:asset-transfer-1}

        \State state = RD        \label{line:state-change-1}

    \EndProcedure
    
    \Procedure{Refund}{Secret $s_{rf}$} \label{line:t-refund}
        \State requires(state == P and IsRefundable(Secret $s_{rf}$))
        \label{line:refund_requirement}
        \State transfer a to s\label{line:asset-transfer-2}
        \State state = RF\label{line:state-change-2}
    \EndProcedure
    
    \Procedure{IsRedeemable}{Secret $s_{rd}$} \label{line:t-isredeemable}
        \State return verify(rd, $s_{rd}$)
    \EndProcedure
    
    \Procedure{IsRefundable}{Secret $s_{rf}$} \label{line:t-isrefundable}
        \State return verify(rf, $s_{rf}$)
    \EndProcedure
\end{algorithmic}
\}
\end{algorithm}

In addition, each smart contract has a redeem function (Line~\ref{line:t-redeem}) and a 
refund function (Line~\ref{line:t-refund}). A redeem function requires the smart 
contract to be in state \textit{P} and that the provided commitment scheme secret is 
valid (Line~\ref{line:redeem_requirement}). If all these requirements hold, the asset $a$ is 
transferred from the contract to the recipient and the contract state is changed to 
\textit{RD} (Lines~\ref{line:asset-transfer-1}--\ref{line:state-change-1}). However, if any requirement is violated, the redeem function fails and 
the smart contract state is not changed.

Similarly, the refund function requires the smart contract to be in \textit{P}
state and that the provided commitment scheme secret is valid (Line~\ref{line:refund_requirement}).
If all these requirements hold, the asset $a$ is refunded from the contract to the sender
and the contract state is changed to \textit{RF} (Lines~\ref{line:asset-transfer-2}--\ref{line:state-change-2}).

The redeem and the refund functions use \textit{two} helper functions: IsRedeemable
(Line~\ref{line:t-isredeemable}) and IsRefundable (Line~\ref{line:t-isrefundable}).
IsRedeemable verifies that the provided redemption commitment scheme secret
is valid and hence the smart contract can be redeemed. Similarly, IsRefundable
verifies that the provided refund commitment scheme secret
is valid and hence the smart contract can be refunded. In 
Section~\ref{sec:solution}, we instantiate different versions of these
two functions for every atomic cross-chain commitment protocol.

%% file: solution.tex
\section{\protocolname: Atomic Cross-Chain Commitment} \label{sec:solution}

This section presents two \textbf{A}tomic \textbf{C}ross-\textbf{C}hain \textbf{C}ommitment,
\protocolname, protocols that achieve both \textbf{atomicity} and \textbf{commitment}
of an \transactionname. There are \textit{two} main challenges in designing a correct 
\protocolname protocol. The first challenge is how to implement the redemption and 
refund commitment scheme instances used by every smart contract in \transactionname. 
The second challenge is how to ensure that the two instances are \textit{mutually exclusive}. 
If the secret of one instance is revealed, the secret of the other instance
\textit{must never} be revealed. 
First, we present AC$^3$TW,
an \protocolname protocol that uses a centralized \textbf{T}rusted  
\textbf{W}itness in Section~\ref{sub:centralized-witness}. Then, we present
\protocolwitness, an \protocolname protocol that replaces the 
centralized trusted witness with a permissionless \textbf{W}itness \textbf{N}etwork 
in Section~\ref{sub:permissionless-witness}. Using
a permissionless network of witnesses does not require more trust in the 
witness network than the required trust in the blockchains used to exchange 
the assets in an \transactionname.  Furthermore, the \protocolwitness protocol 
overcomes the vulnerability of the centralized trusted witness, which may fail or be 
subject to denial of service attacks. 
%Finally, we explain how to implement \protocolname with a 
%decentralized permissioned network of witnesses in 
%Section~\ref{sub:permissioned-witness}

\subsection{AC$^3$TW: Centralized Trusted Witness} \label{sub:centralized-witness}

A centralized trusted witness, Trent, is leveraged to implement an 
\protocolname protocol as follows. For every \transactionname, a directed 
graph $\mathcal{D} = (\mathcal{V},\mathcal{E})$ is constructed at some timestamp $t$ and 
multisigned by all the participants in the set $\mathcal{V}$ generating a graph multisignature
$ms(\mathcal{D})$ as shown in Equation~\ref{equation:1}.
The timestamp $t$ is important to distinguish between identical $AC^{2}T$s among the 
same participants. The order of participant  signatures in $ms(\mathcal{D})$ is 
not important. Any signature order indicates that all participants
in the \transactionname agree on the graph $\mathcal{D}$ at some timestamp $t$.

\begin{equation}
\centering
ms(\mathcal{D}) = sig(...,sig((\mathcal{D}, t), p_1),...,p_{|\mathcal{V}|})
\label{equation:1}
\end{equation}

Afterwards, any participant in the set $\mathcal{V}$ registers $ms(\mathcal{D})$ at Trent
through a registration message. This indicates that participants in the \transactionname 
trust Trent to witness their \transactionname. Trent's identity is leveraged to implement 
the redemption and refund commitment scheme instances and Trent's digital 
signatures~\cite{rivest1978method} are leveraged to implement their corresponding 
commitment scheme secrets.
After $ms(\mathcal{D})$ is registered at Trent,
participants in the \transactionname publish smart contracts in the \transactionname
in their corresponding blockchains. Participants set both the redemption
and refund commitment scheme instances of their smart contracts to the pair 
$(ms(\mathcal{D}), PK_T)$ where $PK_T$ is Trent's public key. Trent's signatures 
$T(ms(\mathcal{D}), RD)$ and $T(ms(\mathcal{D}), RF)$
are used to implement the redemption and refund commitment scheme secrets respectively.

\textbf{The witness role:} Trent maintains a key/value store of 
$ms(\mathcal{D})$'s as the key, and his digital signature to 
either $(ms(\mathcal{D}),RD)$ or $(ms(\mathcal{D}),RF)$ as the value. 
%as shown in Figure~\ref{fig:strawman-solution}. 
When Trent receives a multisigned graph's 
registration message $ms(\mathcal{D})$, Trent checks that $ms(\mathcal{D})$
has not been registered before. If true, Trent inserts $ms(\mathcal{D})$ to the
key/value store. Trent sets the key to $ms(\mathcal{D})$ and sets its 
corresponding value to $\bot$.
%(e.g., $ms_3(\mathcal{D})$ and $ms_4(\mathcal{D})$ in 
%Figure~\ref{fig:strawman-solution}). 
Now, if all participants deploy
their smart contracts in their corresponding blockchains, any participant
can request a redemption signature from Trent for $ms(\mathcal{D})$
through a redemption request message. On receiving a redemption request,
Trent verifies that $ms(\mathcal{D})$ is registered in the key/value store.
If the value corresponding to $ms(\mathcal{D})$ is $\bot$, Trent verifies
that all smart contracts in the \transactionname are deployed and that the redemption
and refund commitment scheme instances of every smart contract are set to 
$(ms(\mathcal{D}), PK_T)$. If true, Trent \textit{witnesses} the redemption of
the \transactionname by signing $(ms,RD)$ and setting the value 
of the key $ms(\mathcal{D})$ to $T(ms(\mathcal{D}),RD)$. However, if the verification
fails, Trent keeps the value of the key $ms(\mathcal{D})$ unchanged.
Similarly, a participant can request a refund signature from Trent 
for $ms(\mathcal{D})$ through a refund request message. On receiving a refund
request, Trent verifies that $ms(\mathcal{D})$ is registered in the key/value store
and its corresponding value is $\bot$. If true, Trent \textit{witnesses} the refund
of the \transactionname by signing $(ms,RF)$ and setting the value of the key 
$ms(\mathcal{D})$ to $T(ms(\mathcal{D}),RF)$. However, 
if the verification fails, Trent keeps the value of the key 
$ms(\mathcal{D})$ unchanged. Trent responds to redemption and refund
requests of $ms(\mathcal{D})$ with the value corresponding to $ms(\mathcal{D})$
in the key/value store. Trent uses the key/value store to ensure that 
either $T(ms(\mathcal{D}), RD)$ or $T(ms(\mathcal{D}), RF)$ can be issued for an \transactionname. 
Once a signature of one commitment scheme secret is revealed (e.g., the redemption signature), 
the signature of the other commitment scheme secret cannot be issued (e.g., the refund 
signature). 
This guarantees that the redemption and refund commitment scheme secrets
are \textbf{mutually exclusive}. 
Trent's signature $T(ms(\mathcal{D}),RD)$ implies that Trent \textit{witnessed} 
the deployment
of all smart contracts in the \transactionname and hence the \transactionname
can be committed without violating atomicity. The \transactionname is committed
once Trent issues $T(ms(\mathcal{D}),RD)$. Afterwards, participants use $T(ms(\mathcal{D}),RD)$ to eventually
redeem all the locked assets in the published smart contracts. Trent's signature 
$T(ms(\mathcal{D}),RF)$ implies that the \transactionname was not previously committed and 
hence can be aborted. The \transactionname is aborted
once Trent issues $T(ms(\mathcal{D}),RF)$. Afterwards, participants use $T(ms(\mathcal{D}),RF)$ to eventually
refund all the locked assets in the published smart contracts.

\begin{comment}
\begin{figure}[ht!]
	\centering
    \includegraphics[width=\columnwidth,height=3cm,keepaspectratio]{img/centralized}
    \caption{Trent maintains a key/value store of multisigned \transactionname graphs as the 
    key and his signature as the value.}
    \label{fig:strawman-solution}
\end{figure}

In the example shown in Figure~\ref{fig:atomic-swap-graph}, both Alice and Bob 
multisign the $AC^{2}T$'s graph $\mathcal{D}$ to generate either $A(B(\mathcal{D}, t))$ or 
$B(A(\mathcal{D}, t))$.  Both $ms = A(B(\mathcal{D}, t))$ and $ms = 
B(A(\mathcal{D}, t))$ indicate that both Alice and Bob agreed on $\mathcal{D}$ at some 
time stamp $t$.  
\end{comment}

\begin{algorithm}[!h]
\caption{Smart contract for centralized \protocolname}
 \label{algo:centralize_smart_contract}

class CentralizedSC extends AtomicSwapSC \{

\begin{algorithmic}[1]
    \Procedure{Constructor}{Address r, MS ms, PK $PK_T$} \label{line:constructor}
        \State this.rd = this.rf = ($ms(\mathcal{D})$, $PK_T$) \label{line:rd_cs}
        \State super(r, this.rd, this.rf) // parent constructor
    \EndProcedure
    
    \Procedure{isRedeemable}{Signature $s_{rd}$} \label{line:isredeemable}
        \State return SigVerify((rd.$ms(\mathcal{D}$), RD), rd.$PK_T$, $s_{rd}$) \label{line:verify-s-rd}
    \EndProcedure
    
    \Procedure{isRefundable}{Signature $s_{rf}$} \label{line:isrefundable}
        \State return SigVerify((rf.$ms(\mathcal{D}$), RF), rf.$PK_T$, $s_{rf}$)
        \label{line:verify-s-rf}
    \EndProcedure
    
\end{algorithmic}
\}
\end{algorithm}

Algorithm~\ref{algo:centralize_smart_contract} presents a smart contract
class inherited from the smart contract template in 
Algorithm~\ref{algo:smart_contract_template} that uses
Trent's digital signatures as redemption and refund commitment 
scheme secrets. Both the redemption and the refund commitment scheme 
instances comprise the ordered pair ($ms(\mathcal{D})$, $PK_T$) (Line~\ref{line:rd_cs}). 
The IsRedeemable function (Line~\ref{line:isredeemable}) takes a 
digital signature as input and verifies that this signature is 
Trent's signature to $(ms(\mathcal{D}), RD)$ using a SigVerify function (Line~\ref{line:verify-s-rd}). 
Similarly, The IsRefundable function (Line~\ref{line:isrefundable}) 
takes a digital signature as input and verifies that this signature is 
Trent's signature to $(ms(\mathcal{D}), RF)$ 
(Line~\ref{line:verify-s-rf}).

The following steps summarizes the AC$^3$TW protocol steps to execute the 
\transactionname shown in Figure~\ref{fig:atomic-swap-graph}:
\begin{enumerate}
    \item Alice and Bob construct the graph $\mathcal{D}$ and multisign
    $(\mathcal{D},t)$ to generate $ms(\mathcal{D})$.
    \item Either Alice or Bob registers $ms(\mathcal{D})$ at Trent and Trent 
    inserts $ms(\mathcal{D})$ to his key/value store only if $ms(\mathcal{D})$ is
    not registered before.
    \item Afterwards, Alice publishes a smart contract $SC_1$ using Algorithm~\ref{algo:centralize_smart_contract}
    to the Bitcoin network stating the following: \label{item:Alice}
    \begin{itemize}
        \item Move X bitcoins from Alice to Bob if Bob provides $T(ms(\mathcal{D}), RD)$.
        \item Refund X bitcoins from $SC_1$ to Alice if Alice provides 
        $T(ms(\mathcal{D}), RF)$.
    \end{itemize}
    \item Concurrently, Bob published a smart contract $SC_2$ to the Ethereum 
    network using Algorithm~\ref{algo:centralize_smart_contract} stating the following: \label{item:Bob}
    \begin{itemize}
        \item Move Y ethers from Bob to Alice if Alice provides $T(ms(\mathcal{D}), RD)$.
        \item Refund Y ethers from $SC_2$ to Bob if Bob provides $T(ms(\mathcal{D}), RF)$.
    \end{itemize}
    \item After both $SC_1$ and $SC_2$ are published, either Alice or Bob requests a 
    redemption commitment scheme secret from Trent. Trent issues $T(ms(\mathcal{D}), RD)$ 
    only if both $SC_1$ and $SC_2$ are published in their corresponding blockchains
    and the value corresponding to $ms(\mathcal{D})$ is $\bot$ in Trent's
    key/value store.
    
    \item If a participant declines to publish their smart contract or a participant
    changes their mind before \transactionname is committed, any
    participant can requests a refund commitment scheme secret from Trent. 
    Trent issues $T(ms(\mathcal{D}), RD)$ only if the value 
    corresponding to $ms(\mathcal{D})$ is $\bot$ in Trent's
    key/value store.
\end{enumerate}

This protocol achieves atomicity by ensuring that either $T(ms(\mathcal{D}), RD)$ or 
$T(ms(\mathcal{D}), RF)$ can be issued. However, this solution requires the participants 
to trust a centralized intermediary, Trent, and hence risks to derail the whole idea of 
blockchain's trust-free decentralization~\cite{nakamoto2008bitcoin}. In 
Section~\ref{sub:permissionless-witness}, we explain how to replace Trent with a permissionless network of witnesses.

\subsection{\protocolwitness: Permissionless Witness Network}\label{sub:permissionless-witness}

This section presents \protocolwitness, an \protocolname protocol that 
uses a \textit{permissionless blockchain network} of witnesses
to decide whether an $AC^{2}T$ should be committed or aborted.
\textbf{Miners} of this blockchain are collectively the \textbf{witnesses} on $AC^{2}T$s. The \protocolwitness protocol is designed to address the shortcoming of the AC$^3$TW 
protocol that depends on a centralized trusted witness. The AC$^3$TW
protocol uses the trusted witness identity and signatures to implement
the redemption and the refund commitment scheme instances of all smart
contracts in the \transactionname. In contrast, in \protocolwitness, it is
infeasible to use a specific witness identity to implement the redemption
and the refund commitment scheme instances. The witness network is
permissionless. Therefore, the identities of all the miners, the witnesses, in 
this network are not necessarily known and hence cannot be used to implement
these commitment scheme instances. Instead, when a set of participants want 
to execute an \transactionname, they deploy a smart contract $SC_w$ in the 
witness network where $SC_w$ is used to coordinate the \transactionname. 
$SC_w$ has a state that determines the state of the \transactionname. 
$SC_w$ exists in one of three states: \textit{Published} ($P$), 
\textit{Redeem\_Authorized} ($RD_{auth}$), or \textit{Refund\_Authorized} 
($RF_{auth}$). Once $SC_w$ is deployed, $SC_w$ is initialized to the 
state $P$. If the witness network decides to commit the \transactionname, the 
witnesses set $SC_w$'s state to $RD_{auth}$. However, if the witness network 
decides to abort the \transactionname, the witnesses set $SC_w$'s state to $RF_{auth}$.

\begin{figure}[ht!]
	\centering
    \includegraphics[width=0.8\columnwidth]{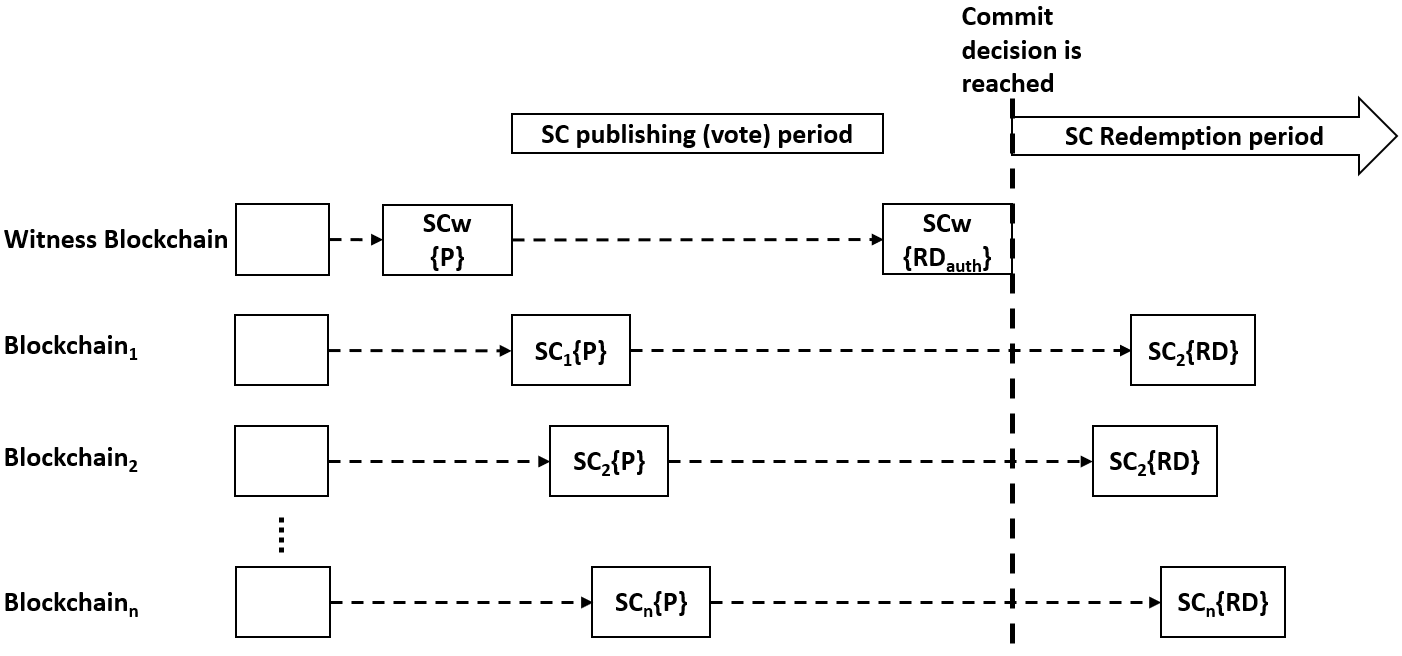}
    \caption{Coordinating $AC^{2}T$s using a permissionless witness network.}
    \label{fig:permissionless-solution}
\end{figure}

Figure~\ref{fig:permissionless-solution} shows an \transactionname that 
exchanges assets among blockchains, $blockchain_1, ..., blockchain_n$
and uses a \textit{witness blockchain} for coordination. Also, it illustrates the 
\protocolwitness protocol steps. For every \transactionname, all 
the participants construct the directed graph $\mathcal{D}$ at some 
timestamp $t$ and multisign it resulting in the multisignature $ms(\mathcal{D})$. A participant registers 
$ms(\mathcal{D})$ in a smart contract $SC_w$ in the witness network where $SC_w$'s state
is initialized to $P$. 
The state $P$ indicates that the participants of the \transactionname agreed on 
$\mathcal{D}$. In addition, the participants agree to conditionally link the 
redeem and the refund actions of their smart contracts in the \transactionname to 
$SC_w$'s states $RD_{auth}$ and $RF_{auth}$ respectively.  Afterwards,
the participants \textit{parallelly} deploy their smart contracts in the blockchains, $blockchain_1, ..., blockchain_n$, as shown
in Figure~\ref{fig:permissionless-solution}. After all the participants
deploy their smart contracts in the \transactionname, a participant may submit
a state
change request to the witness network miners to alter $SC_w$'s state from $P$ to $RD_{auth}$. This request is accompanied by evidence that all smart contracts
in the \transactionname are deployed and correct. Upon receiving this request,
the miners of the witness network verify that $SC_w$'s state is $P$ and that the participants 
of the \transactionname  have indeed deployed their smart contracts in the 
\transactionname in their corresponding blockchains. In addition, the miners
verify that all these smart contracts are in state $P$ and that the redemption and the refund of these smart contracts are conditioned on
$SC_w$'s states $RD_{auth}$ and $RF_{auth}$ respectively. If this verification succeeds,
the miners of the witness network record $SC_w$'s state change to $RD_{auth}$ in their current block.
Once a block that reflects the state change of $SC_w$ to $RD_{auth}$ is mined
in the witness network, the commitment
of the \transactionname is decided and participants can use this block as a commitment
evidence to redeem their assets in the smart contracts of the \transactionname. The 
commit decision is illustrated in Figure~\ref{fig:permissionless-solution} 
using the vertical dotted line.

Similarly, if some participants decline to deploy their smart contracts in 
the \transactionname or a participant changes her mind before the commitment
of the \transactionname, a participant can submit a state
change request to the witness network miners to alter $SC_w$'s state from $P$
to $RF_{auth}$. The miners of the witness network only verify that $SC_w$'s state is $P$. If 
this verification succeeds, the miners of the witness network record $SC_w$'s state change
to $RF_{auth}$ in their current block. Once a block that reflects the state 
change of $SC_w$ to $RF_{auth}$ is mined in the witness network, the 
\transactionname is decided to abort and the participants can use this block as
evidence of the abort to refund their assets in the deployed smart contracts of the \transactionname. 
Note that $SC_w$ is programmed to ensure that $SC_w$'s state can only be changed 
either from $P$ to $RD_{auth}$ or from $P$ to $RF_{auth}$ but no other state 
transition is allowed. This ensures that $SC_w$'s states 
$RD_{auth}$ and $RF_{auth}$ are mutually exclusive. The details of evidence and how 
miners of one blockchain validate evidence in another blockchain without maintaining
a copy of this other blockchain are explained in Section~\ref{sub:evidence}.
Section~\ref{sub:evidence} presents different implementations for evidence submission and
validation.

Algorithm~\ref{algo:permissionless-witness} presents the details
of $SC_w$. $SC_w$ consists of \textit{four} functions: Constructor (Line~\ref{line:w-constructor}),
AuthorizeRedeem (Line~\ref{line:w-authorize-redeem}), AuthorizeRefund (Line~\ref{line:w-authorize-refund}),
and VerifyContracts (Line~\ref{line:w-validate_contracts}). The Constructor initializes $SC_w$
with the participants public keys and the multisigned graph of the \transactionname. This information
is necessary to the witness network miners to later verify the publishing of all smart contracts
in the \transactionname. AuthorizeRedeem alters $SC_w$'s state from $P$ to $RD_{auth}$. 
To call AuthorizeRedeem, a participant provides evidence that shows where the smart
contracts of the \transactionname are published (Line~\ref{line:w-authorize-redeem}).
AuthorizeRedeem first verifies that $SC_w$'s state is currently $P$. In addition,
AuthorizeRedeem verifies that all smart contract in the \transactionname are published 
and correct through a VerifyContracts function call (Line~\ref{line:w-ar-requirements}).
If this verification succeeds, $SC_w$'s state is altered to $RD_{auth}$ 
(Line~\ref{line:w-state-rd}). On the other hand, AuthorizeRefund 
(Line~\ref{line:w-authorize-refund}) verifies only that the state of $SC_w$ is $P$ 
(Line~\ref{line:w-af-requirements}). If true, $SC_w$'s state is altered to $RF_{auth}$ 
(Line~\ref{line:w-state-rf}).

\begin{algorithm}[!h]
\caption{Witness network smart contract as an \transactionname Coordinator.}
 \label{algo:permissionless-witness}

class WitnessSmartContract \{
\begin{algorithmic}[1]
    \State enum State \{Published (P), Redeem\_Authorized ($RD_{auth}$), Refund\_Authorized ($RF_{auth}$)\} \label{line:w-states}
    \State Address [] pk // Addresses of all participants in \transactionname
    \label{line:w-addresses}
    \State Mutlisignature ms // The multisigned graph $\mathcal{D}$
    \label{line:w-multi}
    \State State state
    \label{line:state}
    \Procedure{Constructor}{Address[] pk, MS ms($\mathcal{D})$} \label{line:w-constructor}
        \State this.pk = pk
        \State this.ms = ms($\mathcal{D}$)
        \State this.state = P
    \EndProcedure
    
    \Procedure{AuthorizeRedeem}{Evidence e } \label{line:w-authorize-redeem}
        \State requires (state == P and VerifyContracts(e)) \label{line:w-ar-requirements}
        \State this.state = $RD_{auth}$ \label{line:w-state-rd}
    \EndProcedure
    
    \Procedure{AuthorizeRefund}{} \label{line:w-authorize-refund}
        \State requires (state == P)  \label{line:w-af-requirements}
        \State this.state = $RF_{auth}$ \label{line:w-state-rf}
    \EndProcedure
    
    \Procedure{VerifyContracts}{Evidence e} \label{line:w-validate_contracts}
    \If{e validates all the smart contracts in \transactionname (Check Section~\ref{sub:evidence} for details)}
            \State return true \label{line:w-valid}
    \EndIf
    \State return false \label{line:w-invalid}
    \EndProcedure

\end{algorithmic}
\}
\end{algorithm}

VerifyContracts (Line~\ref{line:w-validate_contracts}) validates that all smart contracts in the \transactionname
are published and correct. For every edge 
$e=(u,v) \in \mathcal{D}.\mathcal{E}$, VerifyContracts finds a matching 
smart contract $SC_e$ in the participant evidence. VerifyContracts ensures that $SC_e$ matches its description in the 
edge $e$. If any parameter in $SC_e$ does not match its description in $e$, VerifyContracts fails and returns 
\textit{false} (Line~\ref{line:w-invalid}). However, if all smart contracts in the provided list are correct,
VerifyContracts returns \textit{true} (Line~\ref{line:w-valid}). VerifyContracts ensures that AuthorizeRedeem cannot be executed 
unless all smart contract in the \transactionname are published and correct and 
hence a commit decision can be reached.

\begin{algorithm}[!h]
\caption{Smart contract for permissionless \protocolname.}
 \label{algo:permissionless_smart_contract}

class PermissionlessSC extends AtomicSwapSC \{

\begin{algorithmic}[1]
    \Procedure{Constructor}{Address r, BC bc, BID bid, TID tid, Depth d} \label{line:permissionless-constructor}
        \State $SC_w$ = retrieveSC(bc, bid, tid)
        \State this.rd = this.rf = ($SC_w$, d)
        \State super(r, this.rd, this.rf) // parent constructor
    \EndProcedure
    
    \Procedure{isRedeemable}{Evidence e} \label{line:permissionless-isredeemable}
        \If{e validates that $SC_w$'s state is $RD_{auth}$ and and that $SC_w$'s state update is at depth $\ge d$}\label{line:permissionless-verify-rd}
            \State return true
        \EndIf
        \State return false
    \EndProcedure
    
    \Procedure{isRefundable}{Evidence e} \label{line:permissionless-isrefundable}
        \If{e validates that $SC_w$'s state is $RF_{auth}$ and that $SC_w$'s state update is at depth $\ge d$}\label{line:permissionless-verify-rf}
            \State return true
        \EndIf
        \State return false
    \EndProcedure
    
\end{algorithmic}
\}
\end{algorithm}

Algorithm~\ref{algo:permissionless_smart_contract}  presents a smart contract
class inherited from the smart contract template in 
Algorithm~\ref{algo:smart_contract_template} in order to use $SC_w$'s state 
as redemption and refund commitment scheme secrets. IsRedeemable
returns \textit{true} if $SC_w$'s state is $RD_{auth}$ (Line~\ref{line:permissionless-verify-rd}), while 
IsRefundable returns \textit{true} if $SC_w$'s state is $RF_{auth}$ 
(Line~\ref{line:permissionless-verify-rf}). As the witness network is permissionless, forks could possibly 
happen resulting in \textit{two} concurrent blocks where $SC_w$'s state is $RD_{auth}$ in the first block and
$SC_w$'s state is $RF_{auth}$ in the second block. To avoid atomicity violations,
participants cannot use a witness network block where $SC_w$'s state is $RD_{auth}$ or 
$RF_{auth}$ in their smart contract redemption and refund respectively unless this block
is buried under at least $d$ blocks in the witness network. As the probability of a fork
of depth $d$ (e.g., 6 blocks in the Bitcoin network~\cite{confirmation}) is negligible,
$SC_w$'s state eventually converges to either $RD_{auth}$ or $RF_{auth}$.

The following steps summarizes the \protocolwitness protocol steps to execute the 
\transactionname shown in Figure~\ref{fig:atomic-swap-graph}:

\begin{enumerate}
    \item Alice and Bob construct the $AC^{2}T$'s graph $\mathcal{D}$ and  
    multisign $(\mathcal{D},t)$ to generate $ms(\mathcal{D})$.
    
    \item Either Alice or Bob registers $ms(\mathcal{D})$ in a smart contract 
    $SC_w$ and publishes $SC_w$ in the witness network setting $SC_w$'s state is $P$.
    $SC_w$ follows Algorithm~\ref{algo:permissionless-witness}.
    \item Afterwards, Alice publishes a smart contract $SC_1$ using Algorithm~\ref{algo:permissionless_smart_contract}
    to the Bitcoin network that states the following: 
    \begin{itemize}
        \item Move X bitcoins from Alice to Bob if Bob provides evidence
        that $SC_w$'s state is $RD_{auth}$.
        \item Refund X bitcoins from $SC_1$ to Alice if Alice provides 
         evidence that $SC_w$'s state is $RF_{auth}$.
    \end{itemize}
    \item Concurrently, Bob publishes a smart contract $SC_2$ to the Ethereum 
    network using Algorithm~\ref{algo:permissionless_smart_contract} stating the following:
    \begin{itemize}
        \item Move Y ethers from Bob to Alice if Alice provides evidence
        that $SC_w$'s state is $RD_{auth}$.
        \item Refund Y ethers from $SC_2$ to Bob if Bob provides 
         evidence that $SC_w$'s state is $RF_{auth}$.
    \end{itemize}
    \item After both $SC_1$ and $SC_2$ are published, any participant can submit
    a state change request of $SC_w$ from $P$ to $RD_{auth}$ to the witness
    network miners. This request is accompanied by evidence that $SC_1$ and $SC_2$ are
    published in the Bitcoin and the Ethereum blockchains respectively.
    The witness network miners first verify that $SC_w$'s state is currently $P$. Then, they verify
    that both $SC_1$ and $SC_2$ are published and correct in their corresponding
    blockchains. If these verifications succeed,
    the miners of the witness network record $SC_w$'s state change to $RD_{auth}$ in their 
    current block. Once a block that reflects the state change of $SC_w$ to 
    $RD_{auth}$ is mined and gets buried under $d$ blocks in the witness network, Alice and Bob can use this
    block as evidence to redeem their assets from $SC_2$ and $SC_1$ 
    respectively.
    
    \item If a participant declines to publish a smart contract, the other participant can submit a state change request of $SC_w$ from $P$ to $RF_{auth}$
    to the witness network miners. The witness network miners verify 
    that $SC_w$'s state is currently $P$. If true, miners record $SC_w$'s state 
    change to $RF_{auth}$ in their current block. 
    Once a block that reflects the state change of $SC_w$ to $RF_{auth}$ is 
    mined and gets buried under $d$ blocks in the witness network, Alice and Bob can use this
    block as evidence to refund their assets from $SC_1$ and $SC_2$ 
    respectively.
    
\end{enumerate}

This protocol uses two blockchain techniques to ensure that $SC_w$'s states $RD_{auth}$ and $RF_{auth}$ are 
\textit{mutually exclusive}. First, it uses the smart contract programmable logic to ensure 
that $SC_w$'s state can only be altered from $P$ to $RD_{auth}$ or from $P$ to $RF_{auth}$. Second,
it uses the longest chain fork resolving technique to resolve forks in the witness network
blockchain. This ensures that in the rare case of forking where one fork chain has $SC_w$'s
state of $RD_{auth}$ and another fork chain has $SC_w$'s state of $RF_{auth}$, the fork is eventually
resolved resulting in either $SC_w$'s state is $RD_{auth}$ or $SC_w$'s state is $RF_{auth}$
but not both.

\subsection{Cross-Chain Evidence Validation} \label{sub:evidence}

This section explains different techniques for the miners of one blockchain, \textit{the validators}, to 
validate the publishing and verify the state of a smart contract deployed in another 
blockchain, \textit{the validated}. The \protocolwitness protocol leverages these techniques in \textit{two}
protocol functions: 1) \texttt{VerifyContracts} in 
Algorithm~\ref{algo:permissionless-witness} and 2) \texttt{IsRedeemable/IsRefundable}
in Algorithm~\ref{algo:permissionless_smart_contract}. In \texttt{VerifyContracts}, the miners
of the witness network need to validate the publishing of all smart contracts in the \transactionname in the blockchains where asset transfers occur. In addition, the miners need to verify that the state of all the published contracts is $P$ and that
the redemption and the refund of these smart contracts are conditioned on $SC_w$'s states. 
Finally, the miners need to verify that for every smart contract in the \transactionname, the sender, the recipient, and the asset match the specification 
of its corresponding edge $e \in \mathcal{D}$. Similarly, in 
\texttt{IsRedeemable/IsRefundable}, miners of the blockchains where asset transfers occur need to verify that
$SC_w$'s state is either $RD_{auth}$, in order to execute the redemption of a smart contract or $RF_{auth}$, 
in order to execute the refund of a smart contract. In the former case, the miners of the witness network
are the validators and the asset blockchains are the validated blockchains. In the latter case, the miners 
of the asset blockchains are the validators and the witness blockchain is the validated blockchain.

A \textbf{simple but impractical} solution is to require all the miners of every 
blockchain to serve as validators to all other blockchains. A blockchain validator 
maintains a copy of the validated blockchain and for every new mined block, a validator 
validates the mined block and adds it to its local copy of the validated blockchain. 
If all mining nodes mine one blockchain and validate all other blockchains, mining 
nodes can consult their local copies of these blockchains to validate the publishing and
hence verify the state of any smart contract in any blockchain. If a participant needs 
the miners of the validator blockchain to validate the publishing of a smart 
contract in the validated blockchain, this participant submits evidence that comprises a block id and a 
transaction id of the smart contract in the validated blockchain to the miners of the validator blockchain. 
This evidence is easily verified by the mining node of the validator blockchain by consulting their copy 
of the validated blockchain. However, this full replication of all the blockchains in all 
the mining nodes is impractical. Not only does it require massive 
processing power to validate all blockchains, but also it requires 
significant storage and network capabilities at each mining node.

Alternatively, miners of one blockchain can run \textit{light 
nodes}~\cite{buterin2014next} of other blockchains. A light node, as 
defined in~\cite{buterin2014next}, is a node that downloads only the block
headers of a blockchain, verifies the proof of work of these block 
headers, and downloads only the blockchain branches that are associated 
with the transactions of interest to this node. For example, the mining
nodes of the witness network can run light nodes for the Bitcoin network to verify the $AC^{2}T$'s smart
contracts in the Bitcoin network. A participant who wants the witness network mining 
nodes to validate a smart contract in the Bitcoin network submits evidence that 
consist of a block id and a transaction id of the smart contract in the Bitcoin network 
to the miners of the witness network. The miners of the witness network use their 
Bitcoin light nodes to validate the smart contract publishing and 
verify the smart contract state. This solution 
requires miners to mine for one blockchain and maintain light nodes for all other 
blockchains. Although the cost of maintaining a light node is much cheaper than
maintaining a blockchain full copy, running a light node for all blockchains 
does not scale as the number of blockchains increases.

It is important to mention that the previous two techniques put the evidence
validation responsibility of one blockchain on the miners of another blockchain. 
In addition, they require changes in the
current infrastructure by requiring the miners of one blockchain to either maintain a
full copy or a light node of other blockchains.

\begin{figure}[ht!]
	\centering
    \includegraphics[width=0.8\columnwidth]{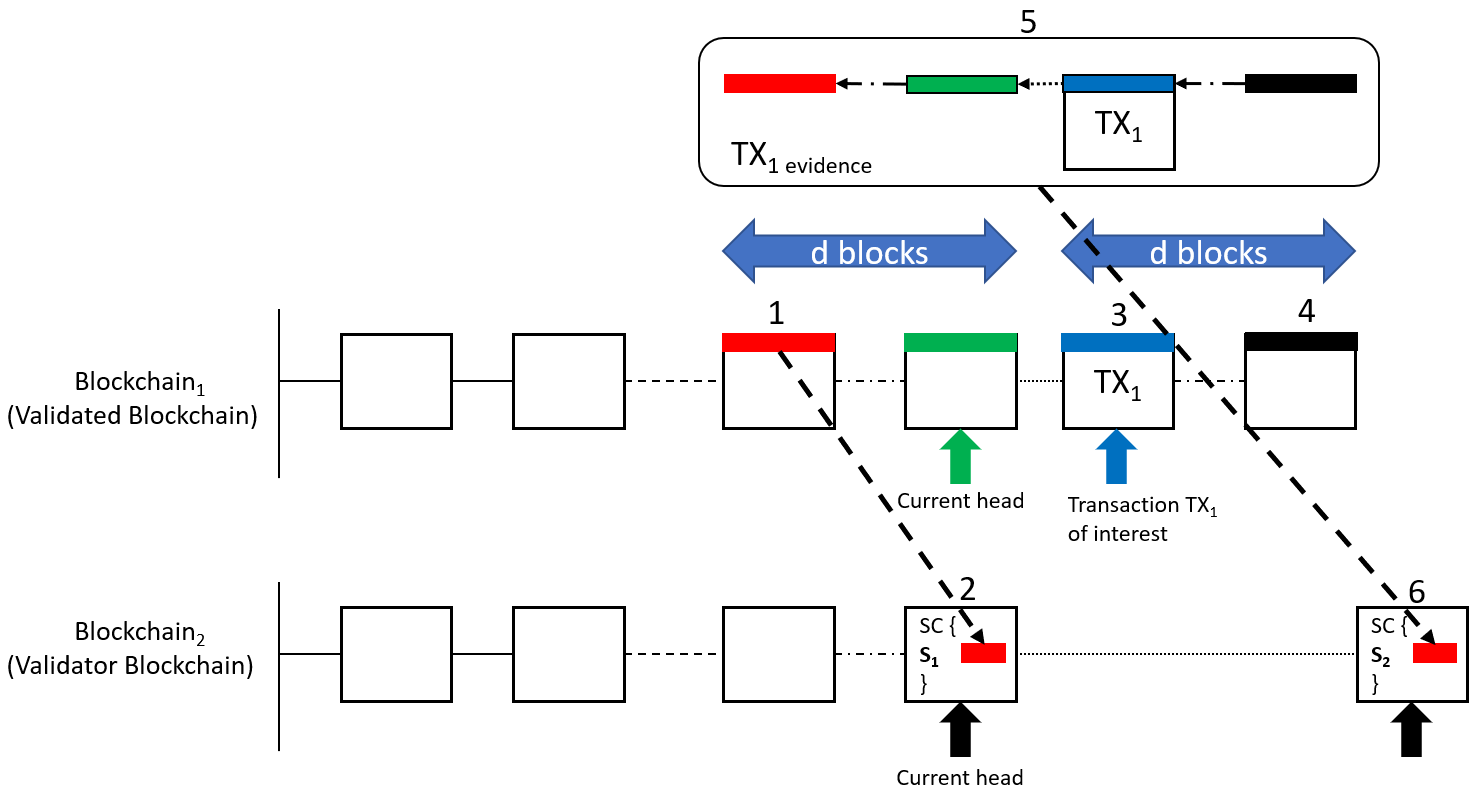}
    \caption{How miners of one blockchain could verify transactions in another blockchain without having
    a copy or a light node of this blockchain.}
    \label{fig:evidence}
\end{figure}

\textbf{Our proposal:} Another way to allow one blockchain, the validator, to validate the publishing and 
verify the state of a smart contract in another blockchain, the validated, is to push the 
validation logic into the code of a smart contract in the validator blockchain. A smart contract in the 
validator blockchain is deployed and stores the header of a \textit{stable block} in the validated 
blockchain. A stable block is a block at depth $d$ from the current head of the validated blockchain such that the 
probability of forking the blockchain at this block is negligible (i.e., a block at depth $\ge$ 6 in the 
Bitcoin blockchain~\cite{confirmation}). A participant who deploys the smart contract in the validator
blockchain stores the block header of a stable block of the validated blockchain as an attribute in the smart 
contract object in the validator blockchain. When the transaction or the smart contract of interest
takes place in a block in the validated blockchain and after this block becomes a stable block, at depth $d$,
a participant can submit evidence of the transaction occurrence in the validated blockchain to the 
miners of the validator blockchain. This evidence comprises the headers of all the blocks that follow the
stored stable block in the smart contract of the validator blockchain in addition to the block 
where the transaction of interest took place. The evidence is submitted to the validator smart contract 
via a function call. This smart contract function validates that the passed headers follow the header of 
the stable block previously stored in the smart contract object and that the proof of work of each header is
valid. In addition, the function verifies that the transaction of interest indeed took place and that the
block of this transaction is stable and buried under $d$ blocks in the validated blockchain.

Figure~\ref{fig:evidence} shows an example of a validator blockchain, blockchain$_2$, that validates the
occurrence of transaction $TX_1$ in the validated blockchain, blockchain$_1$. In this example, there 
exists a smart contract $SC$ that gets deployed in the current head block of blockchain$_2$ (labeled by 
number 2 in Figure~\ref{fig:evidence}). $SC$ has an initial state $S_1$ and stores the header of a stable
block, at depth $d$, in blockchain$_1$ (labeled by number 1). This header is
represented by a red rectangle inside $SC$. $SC$'s state is altered from $S_1$ to $S_2$ if evidence is 
submitted to miners of blockchain$_2$ that proves that $TX_1$ took place in blockchain$_1$ in some block 
after the stored stable block in $SC$. When $TX_1$ takes place in blockchain$_1$ (labeled by number 3) 
and its block becomes a stable block at depth $\ge d$ (labeled by number 4), a participant submits
the evidence (labeled by number 5) to the miners of blockchain$_2$ through $SC$'s function call (labeled
by number 6). This function takes the evidence as a parameter and verifies that blocks in the evidence
took place after the stored stable block in $SC$. This verification ensures that the header of each 
evidence block includes the hash of the header of the previous block starting from the stored stable 
block in $SC$. In addition, this function verifies the proof of work of each evidence's block header. 
Finally, the function validates that $TX_1$ took place in some block in the evidence blocks and that 
this block has already become a stable block. If this verification succeeds, the state of $SC$ is altered from $S_1$ to 
$S_2$. This technique allows miners of one blockchain to verify transactions and smart contracts in 
another blockchain without maintaining a copy of this blockchain. In addition, this technique puts
the evidence validation responsibility on the developer of the validator smart contract.

%% file: generalization.tex
\section{\protocolwitness Analysis} \label{sec:generalizations}

This section analyzes the \protocolwitness protocol introduced
in Section~\ref{sub:permissionless-witness}. First, we establish that 
the proposed protocol ensures atomicity. Then we analyze the scalability of the 
witness network and how it affects the scalability of the commitment protocol. 
Finally, we 
explain how this protocol extends the functionality of previous proposals 
in~\cite{atomicNolan, herlihy2018atomic}.

\subsection{\protocolwitness: Atomicity Correctness Proof} \label{sub:proof}

\begin{lemma}
Assume no forks in the witness network, then the \protocolwitness protocol is atomic.
\end{lemma}
 
\begin{proof}
Assume an \transactionname executed by the \protocolwitness protocol and the atomicity
of this transaction is violated. This atomicity 
violation implies that there exists two smart contract $SC_i$ and $SC_j$ in 
\transactionname where $SC_i$ is redeemed and $SC_j$ is refunded. 
The redemption of $SC_i$ implies that there exists a block in the witness
network where $SC_w$'s state is $RD_{auth}$ while the refund of $SC_j$ implies
that there exists a block in the witness network where $SC_w$'s state is $RF_{auth}$.
Since $SC_w$ is programmed to allow only the state transitions either from $P$ to $RD_{auth}$
or from $P$ to $RF_{auth}$,  the two function calls to alter $SC_w$'s state from $P$ 
to $RD_{auth}$ and from $P$ to $RF_{auth}$ cannot take effect in one block. Miners of
the witness network shall accept one and reject the other. Therefore,
these two state changes must be recorded in two separate blocks. As there exists no 
forks in the witness network, one of these two blocks must
\textit{happen before} the other. This implies that either $SC_w$'s state
is altered from $RD_{auth}$ in one block to $RF_{auth}$ in a following block or altered 
from $RF_{auth}$ in one block to $RD_{auth}$ in a following block. However,
only the state transitions from $P$ to $RD_{auth}$ or from $P$ to $RF_{auth}$
are allowed and no other state transition is permitted leading to a contradiction.
\hfill$\Box$
\end{proof}

\begin{lemma}
Let $\epsilon$ be a negligible probability of forks in the permissionless 
witness network, then \protocolwitness protocol is atomic with a probability $1-\epsilon$.
\end{lemma}
 
\begin{proof}
Assume an \transactionname executed by the \protocolwitness protocol and the atomicity
of this transaction is violated with a probability $p >>> \epsilon$. This atomicity 
violation implies that there exists two smart contract $SC_i$ and $SC_j$ in 
\transactionname where $SC_i$ is redeemed and $SC_j$ is refunded. 
The redemption of $SC_i$ implies that there exists a block in the witness
network where $SC_w$'s state is $RD_{auth}$ while the refund of $SC_j$ implies
that there exists a block in the witness network where $SC_w$'s state is $RF_{auth}$.
As $SC_w$'s states $RD_{auth}$ and $RF_{auth}$ are conflicting states, this 
implies that the block where $SC_w$'s state update to $RD_{auth}$ occurs must exist
in a fork from the block where $SC_w$'s state update to $RF_{auth}$ occurs. The
atomicity violation of the \transactionname with a probability $p$ implies that
the fork probability in the witness network must be $p$ leading to a 
contradiction.
\hfill$\Box$
\end{proof}

\subsection{The Scalability of \protocolwitness}
One important aspect of \protocolname protocols is \textit{scalability}. Does
using a permissionless network of witnesses to coordinate AC$^2$Ts limit
the scalability of the \protocolwitness protocol? In this section, we argue that
the answer is \textit{no}. To explain this argument, we first develop an 
understanding of the properties of executing AC$^2$Ts and the role of the 
witness network in executing AC$^2$Ts.

An \transactionname is a distributed transaction that consists of sub-transactions.
Each sub-transaction is executed in a blockchain. An \protocolname protocol coordinates
the atomic execution of these sub-transactions across several blockchains.
An \protocolname protocol must ensure an 
atomic execution of the distributed transaction. This atomic execution of a distributed transaction requires
the ACID~\cite{gray1981transaction,haerder1983principles} execution of every sub-transaction in this distributed transaction in addition
to the atomic execution of the distributed transaction itself. The ACID execution of a 
sub-transaction executed within a single blockchain is guaranteed by the miners of
this blockchain. Miners use many techniques including mining, verification, and
the miner's rationale to join the longest chain in order to implement ACID executions of
transactions within a single blockchain. The atomicity of the distributed transaction
is the responsibility of the distributed transaction coordinator. Therefore, the main
role of the witness network in the \protocolwitness protocol is to ensure the atomicity 
of the \transactionname. Since the atomicity coordination of AC$^2$Ts is 
\textit{embarrassingly parallel}, different witness network can be used to coordinate
different AC$^2$Ts.

Assume two concurrent AC$^2$Ts, $t_1$ and $t_2$. The atomic execution of $t_1$ 
does not require any coordination with the atomic execution of $t_2$. Each 
\transactionname requires its witness network
to ensure that either all sub-transactions in the \transactionname are 
executed or none of them is executed. Therefore, $t_1$ and 
$t_2$ do not have to be coordinated by the same witness network. 
$t_1$ can be coordinated by one witness network while $t_2$ can be coordinated 
by another witness network. If $t_1$
and $t_2$ conflict at the sub-transaction level, this conflict is resolved by
the miners of the blockchain where these sub-transactions are executed. Therefore,
using a permissionless witness network to coordinate AC$^2$Ts does not limit
the scalability of the \protocolwitness protocol. Different permissionless networks are
used to coordinate different AC$^2$Ts. For example, the Bitcoin network can be used
to coordinate $t_1$ while the Ethereum network can be used to coordinate $t_2$. Once a
performance bottleneck is detected in a permissionless witness network, other
permissionless networks can be potentially used to coordinate other AC$^2$Ts. 
This ensures that the transaction throughput of an \protocolname protocol is only 
bounded by the transaction throughput of the blockchains used to exchange 
the assets in an \transactionname but not the witness network.

\subsection{Handling Complex AC$^2$T Graphs}\label{sub:cyclic-transactions}

\begin{figure}[ht!]
	\centering
    \includegraphics[width=0.7\columnwidth]{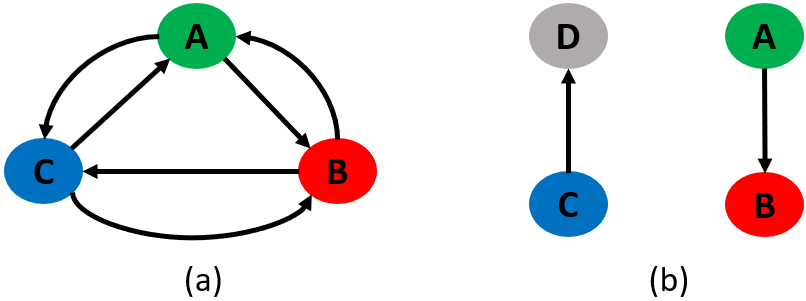}
    \caption{Examples of complex graphs handled by the \protocolwitness protocol: (a) cyclic and (b) disconnected.}
    \label{fig:cyclic}
\end{figure}

One main improvement of the \protocolwitness protocol over the 
state-of-the-art \protocolname protocols 
in~\cite{herlihy2018atomic,atomicNolan} is its ability to coordinate
the atomic execution of AC$^2$Ts with complex graphs.
This improvement is achieved because the \protocolwitness protocol
does not depend on the rational behavior of the participants 
in the \transactionname to ensure atomicity. Instead, the protocol 
depends on a permissionless network of witnesses to coordinate the
atomic execution of AC$^{2}$Ts. Once the participants agree on the 
\transactionname graph and register it in the smart contract $SC_w$ 
in the witness network, participants cannot violate atomicity as the 
commit and the abort decisions are decided by the state of $SC_w$.
The state transitions of $SC_w$ are witnessed and verified by the miners
of the witness network. Therefore, the publishing order of the smart 
contracts in the \transactionname cannot result in an advantage to 
any coalition among the participants. Participants can concurrently 
publish their smart contracts in the \transactionname, both in 
Figures~\ref{fig:atomic-swap-graph} 
and~\ref{fig:cyclic}, without worrying about the maliciousness of any 
participant.

Figure~\ref{fig:cyclic} illustrates two complex graph examples that 
either cannot be atomically executed  by the protocols 
in~\cite{atomicNolan, herlihy2018atomic} or require additional 
mechanisms and protocol modifications to be atomically executed. These
graphs appear in supply-chain applications. Both Nolan's and Herlihy's
single leader protocol require the \transactionname graph to be acyclic
once the leader node is removed. Therefore, both protocols fail to execute the transaction graph shown in Figure~\ref{fig:cyclic}a.
Removing any node from the graph in Figure~\ref{fig:cyclic}a still
results in a cyclic graph. Herlihy presents a multi-leader protocol
in~\cite{herlihy2018atomic} to handle cyclic graphs. However,
both Nolan's and Herlihy's protocols fail to handle disconnected graphs
similar to the graph shown in Figure~\ref{fig:cyclic}b. On the other
hand, the \protocolwitness protocol ensures the 
atomic execution of AC$^2$Ts irrespective of the AC$^2$T's 
graph structure.

%% file: evaluation.tex
\section{Evaluation}\label{sec:evaluation}

This section analytically compares the performance and the overhead of 
the \protocolwitness protocol presented in Section~\ref{sub:permissionless-witness} 
to the state-of-the-art atomic swap protocol presented by Herilhy 
 in~\cite{herlihy2018atomic}. First, we compare the latency of AC$^2$Ts
 as the diameter of the transaction graph $\mathcal{D}$ increases in 
 Section~\ref{sub:latency}. Then, the monetary cost overhead of using a 
 permissionless network of witnesses to coordinate the \transactionname 
 is analyzed in Section~\ref{sub:cost}. Afterwards, an analysis on how
 to choose the witness network is developed in Section~\ref{sub:choosing-witness}.
 Finally, an analysis of the \transactionname throughput as the witness network 
is chosen from the top-4 permissionless cryptocurrencies, sorted by market cap,
is presented in Section~\ref{sub:throughput}.

\subsection{Latency} \label{sub:latency}
The \transactionname latency is defined as the difference between the 
timestamp $t_s$ when an \transactionname is started and the timestamp $t_c$
when the \transactionname is completed. $t_s$ marks the moment when participants 
in the \transactionname start to agree on the \transactionname graph 
$\mathcal{D}$. $t_c$ marks the completion of all the asset transfers in
the \transactionname by redeeming all the smart contracts in \transactionname.

Let $\Delta$ be enough time for any participant to publish a smart contract
in any permissionless blockchain, or to change a smart contract state through
a function call of this smart contract, and for this change to be publicly 
recognized~\cite{herlihy2018atomic}. Also, let $Diam(\mathcal{D})$ be the 
\transactionname graph diameter. The $Diam(\mathcal{D})$ is the length of the
longest path from any vertex in $\mathcal{D}$ to any other vertex in $\mathcal{D}$
including itself.

The single leader atomic swap protocol presented in~\cite{herlihy2018atomic}
has \textit{two} phases: the \transactionname smart contract sequential 
deployment phase and the \transactionname smart contract sequential 
redemption phase. The deployment phase 
requires the deployment of all smart contracts in the \transactionname, $N$, 
where exactly $Diam(\mathcal{D}) \leq N$ smart contracts are sequentially
deployed resulting in a latency of $\Delta \cdot Diam(\mathcal{D})$. Similarly,
the redemption phase requires the redemption of all smart contracts in the 
\transactionname, $N$, where exactly $Diam(\mathcal{D}) \leq N$ smart contracts 
are sequentially redeemed resulting in a latency of $\Delta \cdot Diam(\mathcal{D})$.
The overall latency of an \transactionname that uses this protocol
equals to the latency summation of these two phases 
$2 \cdot \Delta \cdot Diam(\mathcal{D})$. Figure~\ref{fig:latency-herlihy}
visualizes the two phases of the protocol where time advances from left to right. 
As shown, some smart contracts (e.g., $SC_2$, $SC_3$, and $SC_4$) 
could be deployed and redeemed in parallel but
there are exactly $Diam(\mathcal{D})$ sequentially deployed and $Diam(\mathcal{D})$
sequentially redeemed smart contracts resulting in an overall latency of $2 \cdot 
\Delta \cdot Diam(\mathcal{D})$. Note that the protocol allows the parallel
deployment and redemption of some smart contracts as long as they do not
lead to an advantage to either a participant or a coalition in the \transactionname.  

\begin{figure}[ht!]
	\centering
    \includegraphics[width=\columnwidth]{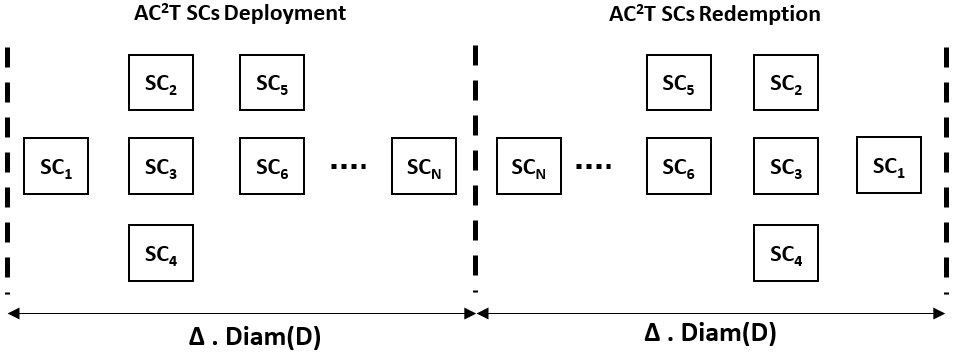}
    \caption{Overall transaction latency of $2 \cdot \Delta \cdot Diam(\mathcal{D})$ when the single leader atomic swap
    protocol in~\cite{herlihy2018atomic} is used.}
    \label{fig:latency-herlihy}
\end{figure}

\begin{comment}

The latency of an \transactionname the uses the single
leader atomic swap protocol in~\cite{herlihy2018atomic} is 
$2.\Delta.Diam(\mathcal{D})$. This latency equals to the latency of the
smart contract sequential deployment phase plus the latency of the 
smart contract sequential redemption phase. The deployment phase requires
the deployment of all smart contracts in the \transactionname, $N$, 
where exactly $Diam(\mathcal{D}) < N$ smart contracts are sequentially
deployed resulting in a latency of $\Delta.Diam(\mathcal{D})$. Similarly,
the redemption phase requires the redemption of all smart contracts in the 
\transactionname, $N$, where exactly $Diam(\mathcal{D}) < N$ smart contracts are 
sequentially redeemed resulting in an overall latency of $\Delta.Diam(\mathcal{D})$.
\end{comment}

On the other hand, the \protocolwitness protocol 
has \textit{four} phases: the witness network smart contract deployment phase, 
the \transactionname smart contract parallel deployment phase, the witness network smart contract state change phase, and the \transactionname smart contract parallel 
redemption phase. The witness network smart contract deployment requires
the deployment of the smart contract $SC_w$ in the witness network resulting in
a latency of $\Delta$. The \transactionname smart contract parallel deployment 
requires the parallel deployment of all smart contracts, N, in the 
\transactionname resulting in a latency of $\Delta$. The witness network smart 
contract state change requires a state change in $SC_w$ either from $P$ to 
$RD_{auth}$ or from $P$ to $RF_{auth}$ through $SC_w$'s Redeem or Refund function 
calls resulting in a latency of $\Delta$. Finally, the \transactionname 
smart contract parallel redemption requires the parallel redemption of all 
smart contracts, N, in the \transactionname resulting in a latency of $\Delta$.
The overall latency of an \transactionname that uses this protocol
equals to the latency summation of these four phases 
$4 \cdot \Delta$. 

\begin{figure}[ht!]
	\centering
    \includegraphics[width=0.6\columnwidth,keepaspectratio]{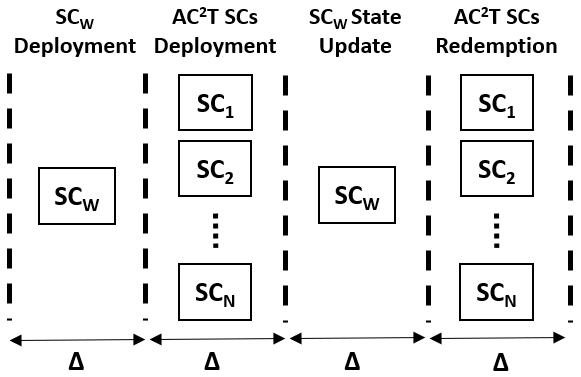}
    \caption{Overall transaction latency of $4 \cdot \Delta$ when the \protocolwitness
    protocol in Section~\ref{sub:permissionless-witness} is used.}
    \label{fig:latency-witness}
\end{figure}

Figure~\ref{fig:latency-witness}
visualizes the four phases of the \protocolwitness protocol where time advances from left to right. 
As shown, all smart contracts in the \transactionname are parallelly deployed and 
parallelly redeemed resulting in an overall latency of $4 \cdot 
\Delta$.

\begin{figure}[ht!]
	\centering
    \includegraphics[width=0.6\columnwidth]{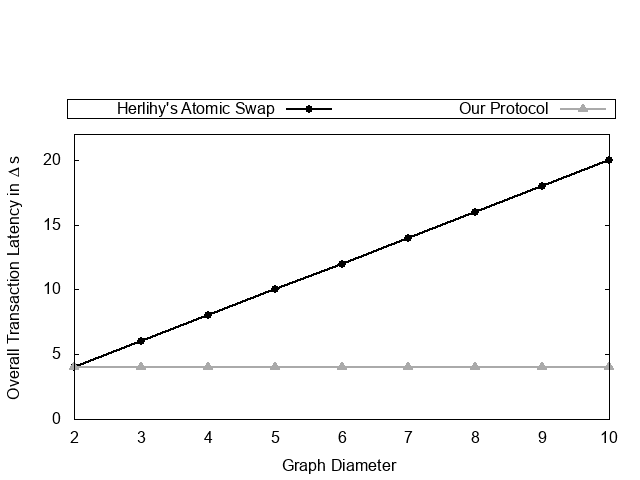}
    \caption{The overall \transactionname latency in $\Delta$s as the graph diameter, $Diam(\mathcal{D})$, increases.}
    \label{fig:overall-transaction-latency}
\end{figure}

Figure~\ref{fig:overall-transaction-latency} compares the overall \transactionname
latency in $\Delta$s resulting from Herlihy's protocol in~\cite{herlihy2018atomic} and the \protocolwitness protocol as the transaction graph diameter, $Diam(\mathcal{D})$
increases. As shown, the \protocolwitness protocol achieves a constant latency of $4 \cdot \Delta$ irrespective
of the transaction graph diameter value while Herlihy's protocol achieves a linear
latency with respect to the transaction graph diameter value. Note that the smallest 
transaction graph consists of \textit{two} nodes and \textit{two} edges
and hence the graph diameter in Figure~\ref{fig:overall-transaction-latency}
starts at 2.

\subsection{Cost Overhead} \label{sub:cost}
This section analyzes the monetary cost overhead of the \protocolwitness protocol 
in comparison to Herlihy's atomic swap protocol in~\cite{herlihy2018atomic}. As explained in Section~\ref{sec:background},
miners charge end-users a fee for every smart contract deployment and every smart contract function
call that results in a smart contract state change. This fee is necessary to incentivize miners
to add smart contracts and append smart contract state changes to their mined
blocks.
As shown in Figures~\ref{fig:latency-herlihy} and~\ref{fig:latency-witness}, both
protocols deploy a smart contract for every edge $e \in \mathcal{E}$ where $\mathcal{E}$ 
is the edge set of the \transactionname graph $\mathcal{D}$. This results in the deployment
of $N = |\mathcal{E}|$ smart contracts in the smart contract deployment phase of both protocols. 
In addition, both 
protocols invoke a redemption or a refund function call for every deployed smart contract 
in the \transactionname resulting in $N$ function calls. However, the
\protocolwitness protocol requires to deploy an additional smart contract $SC_w$ in 
the witness network in addition to an
additional function call to change $SC_w$'s state either from $P$ to $RD_{auth}$
or from $P$ to $RF_{auth}$. The cost of $SC_w$ deployment and $SC_w$ state transition
function call comprises the monetary cost overhead of the 
\protocolwitness protocol. Let $f_d$ be the 
deployment fee of any smart contract $SC_i \in AC^2T$ and $f_{fc}$ be the function
call fee of any smart contract function call. Then, the overall \transactionname
fee of Herlihy's protocol is $N \cdot (f_d + f_{fc})$ while the overall 
\transactionname fee of the \protocolwitness protocol is $(N+1) \cdot (f_d + f_{fc})$.
This analysis shows that \protocolwitness imposes a monetary cost overhead of 
$\frac{1}{N}$ the transaction fee of Herilhy's protocol assuming equal deployment 
and functional call fees for all the smart contracts in the \transactionname.

But, \textit{How much does
it cost in dollars to deploy a smart contract and make a smart contract function call?} 
The answer is, it depends. Many factors affect a smart contract fee such as the length of
the smart contract and the average transaction fee in the smart contract's blockchain~\cite{smartcontractfees, avgtxnfee}. Ryan~\cite{smartcontractfees} shows that
the cost of deploying a smart contract with a similar logic to $SC_w$'s logic in the Ethereum 
network costs approximately \$4 when the ether to USD rate is \$300. Currently,
this costs approximately \$2 assuming the current ether to USD rate of \$140.

\subsection{Choosing the Witness Network} \label{sub:choosing-witness}

This section develops some insights on how to choose the witness network
for an \transactionname. This choice has to consider the risk of 
choosing different permissionless blockchain networks as the witness 
of an \transactionname and the relationship between this risk and the 
value of the assets exchanged in this \transactionname. As the state of 
the witness smart contract $SC_w$ determines the state of an 
\transactionname, forks in the witness network present a risk to
the atomicity of the \transactionname. A fork in the witness network
where one block has $SC_w$'s state of $RD_{auth}$ and another block
has $SC_w$'s state of $RF_{auth}$ might result in an atomicity violation
leading to an asset loss of some participants in the \transactionname.
To overcome possible violation, our \protocolwitness protocol does not
consider a block where $SC_w$'s state is either $RD_{auth}$ or 
$RF_{auth}$ as a commit or an abort evidence until this
block is buried under $d$ blocks in the witness network. This technique
of resolving forks by waiting is presented in~\cite{nakamoto2008bitcoin}
and used by Pass and Shi in~\cite{pass2017hybrid} to eliminate uncertainty
of recently mined blocks. This fork resolution technique is
efficient as the probability of eliminating a fork within 
$d$ blocks is sufficiently high.

However, a malicious participant in an \transactionname could
fork the witness blockchain for $d$ blocks in order to steal
the assets of other participants in the \transactionname. To execute
this attack, a malicious participant rents computing resources
to execute a 51\% attack on the witness network. The cost of an hour of
51\% attack for different cryptocurrency blockchains is presented in
~\cite{51attack}. If the cost of running this attack for $d$ blocks is 
less than the expected gains from running the attack, a malicious
participant is incentivized to act maliciously.

To prevent possible maliciousness, the cost of running a 51\% attack on 
the witness network for $d$ blocks must be set to exceed the potential
gains of running the attack. Let $V_a$ be the value of the potentially
stolen assets if the attack succeeds. Also, let $C_h$ be the hourly
cost of a 51\% attack on the witness network. Finally, let $d_h$ be
 the expected number of mined blocks per hour for the witness blockchain
(e.g., $d_h$ = 6 blocks / hour for the Bitcoin blockchain). The value
$d$ must be set to ensure that $V_a$ is less than the cost of running
the attack for $d$ blocks $\frac{d \cdot C_h}{d_h}$. Therefore
$d$ must be set to achieve the inequality $d > \frac{V_a \cdot d_h}{C_h}$
in order to disincentivize maliciousness. For example, let $V_a$ be \$1M
and assume that the Bitcoin network is used to coordinate this transaction.
The cost per hour of a 51\% attack on the Bitcoin network is 
approximately $C_h =\$300K$. Therefore, $d$ must be set to be $> \frac{\$1M \cdot 
6}{\$300K} = 20$.

\subsection{Throughput} \label{sub:throughput}

The throughput of the AC$^2$Ts is the number of transactions per 
second (tps) that could be processed assuming that every 
\transactionname spans a fixed set of blockchains and is witnessed 
by a fixed witness blockchain. For an \transactionname that spans 
multiple blockchains, the throughput is bounded by the 
slowest involved blockchain in the \transactionname
including the witness network. Let $tps_i$ be the throughput of 
blockchain $i$. The throughput of the AC$^2$Ts
that span blockchains i, j, .., n and are witnessed by
the blockchain w equals to $min(tps_i, tps_j .., tps_n, tps_w$).

\begin{table}[hp!]
\begin{center}
\begin{tabular}{|l|l|l|l|}
\hline
\textbf{Blockchain} & \textbf{tps} & \textbf{Blockchain} & \textbf{tps}\\ \hline
1) Bitcoin             & 7     & 3) Litecoin            & 56       \\ \hline
2) Ethereum            & 25    & 4) Bitcoin Cash        & 61           \\ \hline
\end{tabular}
\caption{The throughput in tps of the top-4 permissionless cryptocurrencies sorted by their market cap~\cite{throughput}.}
\label{table:throughput}
\end{center}
\end{table}

Table~\ref{table:throughput}
shows the transaction throughput of the top-4
permissionless cryptocurrencies sorted by their market cap.
An example AC$^2$T that exchange assets among Ethereum and Litecoin
blockchains and is witnessed by the Bitcoin network achieves
a throughput of 7. The witness network should be chosen from the
set of involved blockchains (Litecoin and Ethereum in this example)
to avoid limiting the transaction throughput.

%% file: conclusion.tex
\section{Conclusion}\label{sec:conclusion}

This paper presents \protocolwitness, the first decentralized \textbf{A}tomic 
\textbf{C}ross-\textbf{C}hain \textbf{C}ommitment protocol that ensures the 
all-or-nothing atomicity semantics even in the presence of participant
crash failures and network delays. Unlike in~\cite{atomicNolan, herlihy2018atomic}
where the protocol correctness mainly relies on participants rational behaviour, 
\protocolwitness separates the coordination of an Atomic Cross-Chain Transaction, \transactionname, from its execution. A permissionless open network of witnesses coordinates
the \transactionname while participants in the \transactionname execute
sub-transactions in the \transactionname. This separation allows \protocolwitness to 
ensure atomicity of all the sub-transactions in an \transactionname even in 
the presence of failures. In addition, this separation enables \protocolwitness to 
parallelly execute sub-transactions in the \transactionname reducing the latency
of an \transactionname from O($Diam(\mathcal{D})$) in~\cite{herlihy2018atomic}, where
$Diam(\mathcal{D})$ is the diameter of the \transactionname graph $\mathcal{D}$,
to $O(1)$ irrespective of the size of the \transactionname graph $\mathcal{D}$.
Also, this separation allows \protocolwitness to scale by using different permissionless
witness networks to coordinate different AC$^2$Ts. This ensures that using a 
permissionless network of witnesses for coordination does not introduce any
performance bottlenecks. Finally, the \protocolwitness protocol extends
the functionality of the protocol in~\cite{herlihy2018atomic} by supporting AC$^2$Ts
with complex graphs (e.g., cyclic and disconnected graphs). 
\protocolwitness introduces a slight monetary cost overhead to the
participants in the \transactionname. This cost equals to the cost of deploying
a coordination smart contract in the witness network plus the cost of a function
call to the coordination smart contract to decide whether to commit or to abort
the \transactionname. The smart contract deployment and function call 
approximately cost \$2\footnote{Assuming 
current ether to USD rate of \$140.} combined per \transactionname
when  the Ethereum network is used to coordinate this \transactionname.